\documentclass[letterpaper,twocolumn,twoside]{IEEEtran}
\usepackage{amsmath, amssymb, cite, epsfig, psfrag}

\def\beq{\begin{equation}}
\def\eeq{\end{equation}}
\def\beqa{\begin{eqnarray}}
\def\eeqa{\end{eqnarray}}
\def\beqan{\begin{eqnarray*}}
\def\eeqan{\end{eqnarray*}}

\def\R{{\mathbb{R}}}

\def\argmin{\mathop{\mathrm{arg\,min}}}
\def\argmax{\mathop{\mathrm{arg\,max}}}
\def\essinf{\mathop{\mathrm{ess\,inf}}}
\def\diag{\mathop{\mathrm{diag}}}
\def\x{\times}

\def\GV{Guo and Verd{\'u}}
\def\MID{\, ; \,}
\newcommand{\Frac}[2]{{{#1}/{#2}}}  % an "inert" form of \frac

\newtheorem{lemma}{Lemma}

\newtheorem{assumption}{Assumption}

\def\xhat{\hat{x}}
\def\SNR{\mbox{\small \sffamily SNR}}

\def\captionSNR{\mbox{\scriptsize \sffamily SNR}}
\def\arr{\rightarrow}
\def\Exp{\mathbf{E}}

\def\ptrue{\ensuremath{p_{0}}}
\def\sigtrue{\ensuremath{\sigma_{0}}}
\def\ppost{\ensuremath{p_{\rm post}}}
\def\sigpost{\ensuremath{\sigma_{\rm post}}}

\def\xbfpmmse{\ensuremath{\hat{\mathbf{x}}^{\rm pmmse}}}
\def\xpmmse{\ensuremath{\hat{x}^{\rm pmmse}}}
\def\xbfmap{\ensuremath{\hat{\mathbf{x}}^{\rm pmap}}}
\def\xmap{\ensuremath{\hat{x}^{\rm pmap}}}
\def\xbflasso{\ensuremath{\hat{\mathbf{x}}^{\rm lasso}}}

\def\xbfzero{\ensuremath{\hat{\mathbf{x}}^{\rm zero}}}
\def\sigEff{\ensuremath{\sigma_{\rm eff}}}
\def\sigEffMap{\ensuremath{\sigma_{\rm eff,map}}}
\def\sigPEff{\ensuremath{\sigma_{\rm p-eff}}}
\def\mse{\mbox{\small \sffamily mse}}
\def\xscaMmse{\ensuremath{\hat{x}^{\rm pmmse}_{\rm scalar}}}
\def\xscaMap{\ensuremath{\hat{x}^{\rm pmap}_{\rm scalar}}}
\def\xscau{\ensuremath{\hat{x}^{u}_{\rm scalar}}}
\def\Tsoft{\ensuremath{T^{\rm soft}}}
\def\Thard{\ensuremath{T^{\rm hard}}}

\def\thetaScau{\ensuremath{\theta^{u}_{\rm scalar}}}
\def\thetau{\ensuremath{\theta^{u}}}
\def\thetaMap{\ensuremath{\theta^{\rm map}}}
\def\thetaScaMap{\ensuremath{\theta^{\rm map}_{\rm scalar}}}

\def\Xset{{\cal X}}

\newcommand{\abf}{\mathbf{a}}
\newcommand{\ubf}{\mathbf{u}}
\newcommand{\ubfhat}{\widehat{\mathbf{u}}}
\newcommand{\wbf}{\mathbf{w}}
\newcommand{\xbf}{\mathbf{x}}
\newcommand{\xbfhat}{\widehat{\mathbf{x}}}
\newcommand{\ybf}{\mathbf{y}}
\newcommand{\thetahat}{\widehat{\theta}}
\newcommand{\Abf}{\mathbf{A}}
\newcommand{\Qbf}{\mathbf{Q}}
\newcommand{\Sbf}{\mathbf{S}}
\newcommand{\Xbf}{\mathbf{X}}

\title{Asymptotic Analysis of MAP Estimation \\ via the Replica Method
       and \\ Applications to Compressed Sensing}
\author{Sundeep Rangan,
        Alyson K. Fletcher,
        and~Vivek~K~Goyal
\thanks{This material is based upon work supported in part by
        a University of California President's Postdoctoral Fellowship and
        the National Science Foundation under CAREER Grant No.\ 0643836\@.}%
\thanks{S. Rangan (email: srangan@poly.edu) is with the Department of
        Electrical and Computer Engineering, Polytechnic Institute of
        New York University.}%
\thanks{A.~K. Fletcher (email: alyson@eecs.berkeley.edu) is with
        the Department of Electrical Engineering and Computer Sciences,
        University of California, Berkeley.}%
\thanks{V.~K. Goyal (email: vgoyal@mit.edu) is with
        the Department of Electrical Engineering and Computer Science and
        the Research Laboratory of Electronics,
        Massachusetts Institute of Technology.}}

\begin{document}

\maketitle

\begin{abstract}
The replica method is a non-rigorous but well-known technique
from statistical physics used in the
asymptotic analysis of large, random, nonlinear problems.
This paper applies the replica method, under the assumption of replica symmetry,
to study estimators that are maximum a posteriori (MAP) under a postulated
prior distribution.
It is shown that with
random linear measurements and Gaussian noise, the replica-symmetric prediction of
the asymptotic behavior of the postulated MAP estimate of an $n$-dimensional
vector ``decouples'' as $n$ scalar postulated MAP estimators.
The result is based on applying a hardening argument to
the replica analysis of postulated posterior mean estimators
of Tanaka and of {\GV}.

The replica-symmetric postulated MAP analysis can be readily applied to many
estimators used in compressed sensing, including
basis pursuit, lasso, linear estimation with thresholding,
and zero norm-regularized estimation.
In the case of lasso estimation the scalar
estimator reduces to a soft-thresholding operator, and for
zero norm-regularized estimation it reduces to a hard-threshold.
Among other benefits, the replica method provides a
computationally-tractable method for precisely predicting various performance
metrics including mean-squared error and sparsity pattern recovery
probability.
\end{abstract}


\section{Introduction}
\label{sec:intro}

Estimating a vector $\xbf \in \R^n$ from measurements of the form
\beq \label{eq:yPhix}
    \ybf = \Phi\xbf + \wbf,
\eeq
where $\Phi \in \R^{m \x n}$
represents a known \emph{measurement matrix} and $\wbf \in \R^m$
represents measurement errors or noise,
is a generic problem that arises in a range of circumstances.
When the noise $\wbf$ is i.i.d.\ zero-mean Gaussian with variance $\sigma^2$
and $\xbf$ is i.i.d.\ with components $x_j$ having a probability distribution function $p(x_j)$,
the \emph{maximum a posteriori} (MAP) estimate is given by
\beq \label{eq:xhatMap-intro}
    \xbfmap(\ybf) = \argmin_{\xbf \in \R^n} \left[ \frac{1}{2\sigma^2}\|\ybf-\Phi\xbf\|^2 +
        \sum_{j=1}^n f(x_j) \right],
\eeq
where $f(x_j) = -\log p(x_j)$.
Estimators of the form \eqref{eq:xhatMap-intro} are also used with the regularization function
$f(x_j)$ or noise level parameter $\sigma^2$ not matching the true prior or noise level,
either since those quantities are not known or since the optimization in \eqref{eq:xhatMap-intro}
using the true values is too difficult to compute.
In such cases, the estimator \eqref{eq:xhatMap-intro} can be interpreted as a MAP estimate
for a \emph{postulated} distribution and noise level, and we will thus call
estimators of the form \eqref{eq:xhatMap-intro} \emph{postulated MAP estimators}.

Due to their prevalence, characterizing the behavior
of postulated MAP estimators is of interest in a wide range of applications.
However, for most regularization functions $f(\cdot)$, the postulated MAP estimator
\eqref{eq:xhatMap-intro} is
nonlinear and not easy to analyze.
Even if, for the purpose of analysis, one assumes separable priors on
 $\xbf$ and $\wbf$, the analysis of the postulated MAP estimate
may be difficult since the matrix $\Phi$ couples the $n$ unknown
components of $\xbf$ with the $m$ measurements in the vector $\ybf$.

This paper provides a general analysis of postulated MAP estimators
based on the \emph{replica method}---a non-rigorous but widely-used method from statistical
physics for analyzing large random systems.
It is shown that, under a key assumption of replica symmetry described below,
the replica method predicts that
with certain large random $\Phi$ and Gaussian $\wbf$,
there is an \emph{asymptotic decoupling}
of the vector postulated MAP estimate \eqref{eq:xhatMap-intro}
into $n$ scalar MAP estimators.
Specifically, the replica method predicts that the
joint distribution of each component $x_j$ of $\xbf$ and its
corresponding component $\xhat_j$ in the estimate vector $\xbfmap(\ybf)$
is asymptotically identical to the outputs of a simple system where $\xhat_j$
is a postulated  MAP estimate of the scalar random variable $x_j$ observed in
Gaussian noise.
Using this scalar equivalent model, one can then readily compute
the asymptotic joint distribution of $(x_j,\xhat_j)$ for any component $j$.

The replica method's non-rigorous but simple prescription for
computing the asymptotic joint componentwise distributions
has three key, attractive features:
\begin{itemize}

\item \emph{Sharp predictions:}  Most importantly,
the replica method provides---under the assumption of the
replica hypotheses---not just bounds, but sharp predictions of the
asymptotic behavior of postulated MAP estimators.
From the joint distribution, various further computations
can be made, to provide precise predictions of quantities such as
the mean-squared error (MSE)
and the error probability of any componentwise hypothesis test computed
from a postulated MAP estimate.

\item \emph{Computational tractability:}  Since the scalar equivalent model
involves only a scalar random variable $x_j$, scalar Gaussian noise,
and scalar postulated MAP estimate $\xhat_j$, any quantity derived from the joint
distribution can be computed numerically from one- or two-dimensional
integrals.

\item \emph{Generality:}
The replica analysis can incorporate arbitrary separable
distributions on $\xbf$ and regularization functions $f(\cdot)$.
It thus applies to a large class of estimators and test scenarios.
\end{itemize}

\subsection{Replica Method and Contributions of this Work}
\label{sec:intro-replica}

The replica method was originally developed by
Edwards and Anderson~\cite{EdwardsA:75}
to study the statistical mechanics of spin glasses.
Although not fully rigorous from the perspective of probability
theory, the technique was able to provide explicit
solutions for a range of complex problems where many other methods
had previously failed.
Indeed, the replica method and related ideas from statistical mechanics
have found success in a number of classic
NP-hard problems including
the traveling salesman problem~\cite{MezardP:86},
graph partitioning~\cite{FuA:86}, $k$-SAT~\cite{MonassonZ:97}
and others~\cite{Nishimori:01}.
Statistical physics methods have also been applied to the study
of error correcting codes~\cite{Sourlas:89,Montanari:00}.
There are now several general texts on the replica method
\cite{Dotsenko:95,MezardParVir:01,MezardM:09,Talagrand:03}.

The replica method was first applied to the
study of nonlinear MAP estimation problems
by Tanaka~\cite{Tanaka:02}.  That work applied what is called a replica symmetric analysis to
multiuser detection for large CDMA systems
with random spreading sequences.
M{\"u}ller~\cite{Muller:03} considered a mathematically-similar problem for
MIMO communication systems.
In the context of the estimation problem considered here,
Tanaka's and M{\"u}ller's papers essentially characterized the behavior of
the MAP estimator of a vector $\xbf$ with i.i.d.\ binary components
observed through linear measurements of the form (\ref{eq:yPhix}) with a
large random $\Phi$ and Gaussian $\wbf$.

Tanaka's results were then
generalized in a remarkable paper by {\GV}~\cite{GuoV:05}
to vectors $\xbf$ with arbitrary separable distributions.
Guo and  Verd{\'u}'s result was also able to incorporate
a large class of postulated minimum mean squared error (MMSE) estimators,
where the estimator may assume a prior that is different from the actual prior.
Replica analyses have also been applied to related communication problems
such as lattice precoding for the Gaussian broadcast channel~\cite{ZaiMuMM:10arXiv}.
A brief review of the replica method analysis by Tanaka~\cite{Tanaka:02} and
{\GV}~\cite{GuoV:05} is provided in Appendix~\ref{sec:replica}.

The result in this paper is derived from {\GV}~\cite{GuoV:05}
by a standard hardening argument.
Specifically, the postulated MAP estimator \eqref{eq:xhatMap-intro}
is first expressed as a limit of the postulated MMSE estimators analyzed in~\cite{GuoV:05}.
Then, the behavior of the postulated MAP estimator can
be derived by taking appropriate limits of the results in~\cite{GuoV:05} on
postulated MMSE estimators.
This hardening technique is well-known and is used in
Tanaka's original work~\cite{Tanaka:02}
in the analysis of MAP estimators with binary and Gaussian priors.

Through the limiting analysis via hardening,
the postulated MAP results here follow from the postulated MMSE results in~\cite{GuoV:05}.
Thus, the central contribution of this work is to work out these limits
to provide a set of equations for a general class of postulated MAP estimators.
In particular, while Tanaka has derived the equations for replica predictions of
MAP estimates for binary and Gaussian priors, the results here provide explicit equations
for general priors and regularization functions.

\subsection{Replica Assumptions}

The non-rigorous aspect of the replica method involves
a set of assumptions that include a self-averaging property,
the validity of a ``replica trick,'' and the ability to exchange
certain limits.  Importantly, this work is based on an additional strong assumption of
\emph{replica symmetry}.  As described in Appendix~\ref{sec:replica},
the replica method reduces the calculation
of a certain free energy to an optimization problem over covariance matrices.  The
replica symmetric (RS)
assumption is that the maxima in this optimization satisfy certain symmetry properties.
This RS assumption is not always valid, and indeed Appendix~\ref{sec:replica}
provides several examples from other applications of the replica method
where replica symmetry breaking (RSB) solutions are known to be needed
to provide correct predictions.

For the analysis of postulated MMSE estimators,
\cite{Tanaka:02} and~\cite{GuoV:05} derive analytic conditions for the
validity of the RS assumption only in some limited cases.
Our analysis of postulated MAP estimators
depends on~\cite{GuoV:05}, and, unfortunately, we have not provided a
general analytic test for the validity of the
RS assumption in this work.
Following~\cite{GuoV:05}, our approach instead is to compare, where possible,
the predictions under the RS assumption to numerical simulations of the
postulated MAP estimator.
As we will see in Section~\ref{sec:sim},
the RS predictions appear to be accurate, at least for many common estimators arising in compressed
sensing.
That being said, the RS analysis can also provide predictions for optimal MMSE
and zero norm-regularized estimators that cannot be simulated tractably.
Extra caution must be applied in assuming
the validity of the RS predictions for these estimators.

To emphasize our dependence on these unproven assumptions---notably replica symmetry---we will
refer to the general MMSE analysis in {\GV}'s work~\cite{GuoV:05} as the
\textbf{replica symmetric postulated MMSE decoupling property}.
Our main result
will be called the \textbf{replica symmetric postulated MAP decoupling property}.

\subsection{Connections to Belief Propagation} \label{sec:bpConn}
Although not explored in this work, it is important to point out that
the results of the replica analysis of postulated MMSE and MAP estimation
are similar to those derived for belief propagation (BP) estimation.
Specifically, there is now a large body of work analyzing BP and approximate
BP algorithms for estimation of vectors $\xbf$ observed through linear
measurements of the form \eqref{eq:yPhix} with large random $\Phi$.
For both certain large sparse random matrices~\cite{BoutrosC:02,TanakaO:05,MontanariT:06,GuoW:06,GuoW:07,Rangan:10-CISS,Rangan:10arXiv},
and more recently
for certain large dense random matrices~\cite{DonohoMM:09,BayatiM:11,Rangan:10carXiv,Rangan:11-ISIT},
several results now show that BP estimates exhibit an asymptotic decoupling property
similar to RS predictions for postulated MMSE and MAP estimators.
Graphical model arguments have also been used to establish a decoupling
property under a very general, random sparse observation model~\cite{Montanari:08arXiv}.

The effective noise level in the scalar equivalent model for BP and approximate BP methods
can be predicted by certain state evolution equations
similar to density evolution analysis of BP decoding of
LDPC codes~\cite{TenBrink:01,AshkiminKT:04}.  It turns out that in several cases, the fixed point
equations for state evolution are identical to the equations for the effective noise level predicted by
the RS analysis of postulated MMSE and MAP estimators.
In particular, the equations in~\cite{DonohoMM:09,BayatiM:11} agree exactly
with the RS predictions for LASSO estimation given in this work.

These connections are significant in several regards:
Firstly, the state evolution analysis of BP algorithms can be made fully rigorous
under suitable assumptions and thus
provides an independent, rigorous justification for some of the RS claims.

Secondly, the replica method provides only an \emph{analysis} of estimators,
but no method to actually compute those estimators.
In contrast, the BP and approximate BP algorithms provide a possible
tractable method for achieving the performance predicted by the replica method.

Finally, the BP analysis provides an algorithmic intuition as to why decoupling may occur
(and hence when replica symmetry may be valid):
As described in~\cite{Montanari:11}, BP and approximate BP algorithms can be seen
as iterative procedures where the vector estimation problem is reduced to a sequence of
``decoupled" scalar estimation problems.
This decoupling is based essentially on
the
principle that, in each iteration, when estimating one component $x_j$, the uncertainty in the other
components $\{x_k,\, k \neq j\}$ can be aggregated as Gaussian noise.
Based on the state evolution analysis of BP algorithms, we know that
this Central Limit Theorem-based approximation is asymptotically valid when the components of the mixing matrix $\Phi$
are sufficiently dense and independent.
Thus, the validity of RS is possibly connected to validity of this Gaussian approximation.

\subsection{Applications to Compressed Sensing}
As an application of our main result, we will develop
a few analyses of estimation problems that arise in
compressed sensing~\cite{CandesRT:06-IT,Donoho:06,CandesT:06}.
In \emph{compressed sensing}, one estimates a sparse vector $\xbf$
from random linear measurements.
A vector $\xbf$ is \emph{sparse} when its number of nonzero entries $k$
is smaller than its length $n$.
Generically, optimal estimation of $\xbf$ with a sparse prior is
NP-hard~\cite{Natarajan:95}.
Thus, most attention has focused on
greedy heuristics such as
matching pursuit~\cite{ChenBL:89,MallatZ:93,PatiRK:93,DavisMZ:94}
and convex relaxations such as basis pursuit~\cite{ChenDS:99} or
lasso~\cite{Tibshirani:96}.
While successful in practice,
these algorithms are difficult to analyze precisely.

Compressed sensing of sparse $\xbf$ through (\ref{eq:yPhix})
(using inner products with rows of $\Phi$) is mathematically identical
to \emph{sparse approximation} of $\ybf$ with respect to columns of $\Phi$.
An important set of results for both sparse approximation and
compressed sensing are the deterministic conditions on the
\emph{coherence} of $\Phi$ that are sufficient to guarantee good performance
of the suboptimal methods mentioned above~\cite{Tropp:04,Tropp:06,DonohoET:06}.
These conditions can be satisfied with high probability
for certain large random measurement matrices.
Compressed sensing has provided many sufficient conditions that are
easier to satisfy than the initial coherence-based conditions.
However, despite this progress, the exact performance
of most sparse estimators is still not known precisely,
even in the asymptotic case of large random measurement matrices.
Most results describe the estimation performance via bounds,
and the tightness of these bounds is generally not known.

There are, of course, notable exceptions including~\cite{Wainwright:09-lasso}
and~\cite{DonohoT:09},
which provide matching necessary and sufficient conditions
for recovery of strictly sparse vectors with basis pursuit and lasso.
However, even these results only consider exact recovery
and are limited to measurements that are noise-free or
measurements with a signal-to-noise ratio (SNR) that scales to infinity.

Many common sparse estimators can be seen as MAP estimators
with certain postulated priors.  Most importantly,
lasso and basis pursuit are MAP estimators assuming a Laplacian prior.
Other commonly-used sparse estimation algorithms, including
linear estimation with and without thresholding
and zero norm-regularized estimators, can also be seen as postulated
MAP-based estimators.
For these postulated MAP-based sparse estimation algorithms, the replica method can
provide non-rigorous but sharp, easily-computable predictions for the
asymptotic behavior.  In the context of compressed sensing, this analysis
can predict various performance metrics such as MSE or fraction of
support recovery.
The expressions can apply to arbitrary ratios $k/n$, $n/m$, and $\SNR$.
Due to the generality of the replica analysis,  the methodology can also
incorporate arbitrary
distributions on $\xbf$ including several sparsity models,
such as Laplacian,
generalized Gaussian, and Gaussian mixture priors.
Discrete distributions can also be studied.

It should be pointed out that this work is not the
first to use ideas from statistical physics for the study
of sparse estimation.
Guo, Baron and Shamai~\cite{GuoBS:09-Allerton} have provided a
replica analysis of compressed sensing that characterizes not just the
postulated MAP or postulated MMSE
estimate, but the asymptotic posterior marginal distribution.  That work
also shows an independence property across finite sets of components.
Merhav, Guo and Shamai~\cite{MerhavGS:10}
consider, among other applications, the estimation of a sparse
vector $\xbf$ from measurements of the form $\ybf = \xbf + \wbf$.
In their model, there is no measurement matrix such as $\Phi$
in (\ref{eq:yPhix}), but the components of $\xbf$ are possibly
correlated.  Their work derives explicit expressions
for the MMSE as a function of the probability distribution
on the number of nonzero components.  The analysis does
not rely on replica assumptions and is fully rigorous.
More recently, Kabashima, Wadayama and Tanaka~\cite{KabashimaWT:09arXiv}
have used the replica method
to derive precise
conditions on which sparse signals can be recovered
with $\ell_p$-based relaxations such as lasso.
Their analysis does not consider noise, but can find conditions
on recovery on the entire vector $\xbf$, not just individual
components.
Also, using free probability theory~\cite{Voiculescu:91,Muller:04},
a recent analysis~\cite{CaiShTuVer:11}
extends the replica analysis of compressed sensing to larger classes of
matrices, including matrices $\Phi$ that are possibly not i.i.d.

\subsection{Outline}
The remainder of the paper is organized as follows.
The precise estimation problem is described in Section~\ref{sec:model}.
We review the RS postulated MMSE decoupling property of {\GV}
in Section~\ref{sec:replicaMMSE}.  We then present our
main result, an RS postulated MAP decoupling property,
in Section~\ref{sec:replicaMAP}.
The results are applied to the analysis of compressed sensing algorithms in
Section~\ref{sec:compSensAnal}, which is followed by numerical
simulations in Section~\ref{sec:sim}.
Conclusions are possible avenues for future work are given in Section~\ref{sec:conclusions}.
The proof of the main result is somewhat long and given in a set of
appendices; Appendix~\ref{sec:proof-overview} provides an overview of
the proof and a guide through the appendices with detailed arguments.

\section{Estimation Problem and Assumptions}
\label{sec:model}

Consider the estimation of a random vector $\xbf \in \R^n$
from linear measurements of the form
\beq \label{eq:yax}
    \ybf = \Phi\xbf + \wbf = \Abf\Sbf^{1/2}\xbf + \wbf,
\eeq
where $\ybf \in \R^m$ is a vector of observations;
$\Phi = \Abf\Sbf^{1/2}$,
with $\Abf \in \R^{m\x n}$, is a measurement matrix;
$\Sbf$ is a diagonal matrix of positive scale factors,
\beq \label{eq:Sdef}
    \Sbf = \diag\left(s_1,\ldots,s_n\right), \ \ \ s_j > 0;
\eeq
and $\wbf \in \R^m$ is zero-mean, white Gaussian noise.
We consider a sequence of such problems indexed by $n$,
with $n \rightarrow \infty$.
For each $n$, the problem is to determine an estimate $\xbfhat$ of $\xbf$
from the observations $\ybf$ knowing the measurement matrix $\Abf$
and scale factor matrix $\Sbf$.

The components $x_j$ of $\xbf$ are modeled as zero mean and i.i.d.\
with some prior probability distribution $\ptrue(x_j)$.
The per-component variance of the Gaussian noise is $\Exp|w_j|^2 = \sigtrue^2$.
We use the subscript ``0'' on the prior and noise level
to differentiate these quantities from certain
``postulated'' values to be defined later.
When we develop applications in Section~\ref{sec:compSensAnal},
the prior $\ptrue(x_j)$
will incorporate presumed sparsity of $\xbf$.

In (\ref{eq:yax}), we have factored $\Phi = \Abf\Sbf^{1/2}$
so that even with the i.i.d.\ assumption on $\{x_j\}_{j=1}^n$ above and an
i.i.d.\ assumption on entries of $\Abf$, the model can capture
variations in powers of the components of $\xbf$
that are known \emph{a priori} at the estimator.
Specifically, multiplication by $\Sbf^{1/2}$
scales the variance of the $j$th component of $\xbf$ by
a factor $s_j$.
Variations in the power of $\xbf$ that are not known
to the estimator should be captured in the distribution of $\xbf$.

We summarize the situation and make additional assumptions
to specify the problem precisely as follows:
\begin{itemize}
\item[(a)] The number of measurements
   $m = m(n)$ is a deterministic quantity that varies with $n$ and satisfies
   \[
       \lim_{n \arr \infty} n/ m(n) = \beta
   \]
   for some $\beta \geq 0$.
   (The dependence of $m$ on $n$ is usually omitted for brevity.)
\item[(b)] The components $x_j$ of $\xbf$ are i.i.d.\ with probability
   distribution $\ptrue(x_j)$.  All moments of $x_j$ are finite.
\item[(c)] The noise $\wbf$ is Gaussian with
   $\wbf \sim {\cal N}(0,\sigtrue^2 \mathbf{I}_m)$.
\item[(d)] The components of the matrix $\Abf$ are i.i.d.\ and distributed as
$A_{ij} \sim (1/\sqrt{m})A$ for some random variable $A$ with zero
mean, unit variance and all other moments of $A$ finite.
\item[(e)] The scale factors $s_j$ are i.i.d., satisfy
   $s_j > 0$ almost surely, and all moments of $s_j$ are finite.
\item[(f)] The scale factor matrix $\Sbf$, measurement matrix $\Abf$,
vector $\xbf$, and noise $\wbf$ are all independent.
\end{itemize}

\section{Review of the Replica Symmetric Postulated MMSE Decoupling Property}
\label{sec:replicaMMSE}
We begin by reviewing the
RS postulated MMSE decoupling property of {\GV}~\cite{GuoV:05}.

\subsection{Postulated MMSE Estimators}
\label{sec:pmmse}
To define the concept of a postulated MMSE estimator,
suppose one is given a ``postulated'' prior distribution $\ppost$
and a postulated noise level $\sigpost^2$ that may be different
from the true values $\ptrue$ and $\sigtrue^2$.
We define the \emph{postulated minimum MSE} (PMMSE) estimate of $\xbf$ as
\beqa
    \xbfpmmse(\ybf) &=& \Exp\left( \xbf \mid \ybf \MID \ppost, \sigpost^2 \right) \nonumber \\
    &=& \int \xbf p_{\xbf \mid \ybf}(\xbf \mid \ybf \MID \ppost, \sigpost^2) \, d\xbf,
        \label{eq:xhatMmse}
\eeqa
where $p_{\xbf \mid \ybf}(\xbf \mid \ybf \MID q, \sigma^2)$ is the conditional
distribution of $\xbf$ given $\ybf$
under the $\xbf$ distribution $q$ and noise variance $\sigma^2$
specified as parameters after the semicolon.
We will use this sort of notation throughout the rest of the paper,
including the use of $p$ without a subscript for the p.d.f.\ of the scalar
or vector quantity understood from context.
In this case, due to the Gaussianity of the noise, we have
\beqa
    \lefteqn{p_{\xbf \mid \ybf}(\xbf \mid \ybf \MID q, \sigma^2) } \nonumber \\
    &=& C^{-1} \exp\left(-\frac{1}{2\sigma^2}\|\ybf-\Abf\Sbf^{1/2}\xbf\|^2\right)
    q(\xbf),
    \label{eq:pxy}
\eeqa
where the normalization constant is
\[
    C = \int \exp\left(-\frac{1}{2\sigma^2}
        \|\ybf-\Abf\Sbf^{1/2}\xbf\|^2\right)q(\xbf) \, d\xbf
\]
and
$q(\xbf)$ is the joint p.d.f.\
\[
    q(\xbf) = \prod_{j=1}^n q(x_j).
\]

In the case when $\ppost = \ptrue$
and $\sigpost^2 = \sigtrue^2$, so that the postulated and true
values agree, the PMMSE estimate reduces to the true MMSE estimate.

\subsection{Decoupling under Replica Symmetric Assumption}
\label{sec:pmmse-decoupling}
The essence of the RS PMMSE decoupling property
is that the asymptotic behavior of the PMMSE estimator
is described by an equivalent scalar estimator.
Let $q(x)$ be a probability distribution defined
on some set $\Xset \subseteq \R$.  Given $\mu > 0$,
let $p_{x\mid z}(x \mid z \MID q, \mu)$ be the conditional distribution
\beqa
    \lefteqn{p_{x\mid z}(x \mid z \MID q,\mu)} \nonumber \\
     &=& \left[ \int_{x \in \Xset} \phi(z-x \MID \mu)q(x) \, dx \right]^{-1}
     \phi(z-x \MID \mu)q(x)
    \label{eq:pxz}
\eeqa
where $\phi(\cdot)$ is the Gaussian distribution
\beq \label{eq:Gaussian}
    \phi(v \MID \mu) = \frac{1}{\sqrt{2\pi\mu}}
        e^{-\Frac{|v|^2}{(2\mu)}}.
\eeq
The distribution $p_{x\mid z}(x | z \MID q,\mu)$
is the conditional distribution of
the scalar random variable $x \sim q(x)$
given an observation of the form
\beq \label{eq:zmse}
    z = x + \sqrt{\mu}v,
\eeq
where $v \sim {\cal N}(0,1)$.
Using this distribution, we can
define the scalar conditional MMSE estimate
\beq \label{eq:xmmseScalar}
    \xscaMmse(z \MID q, \mu) = \int_{x \in \Xset} x \,
    p_{x \mid z}(x \mid z \MID \mu) \, dx.
\eeq
Also, given two distributions, $p_0(x)$ and $p_1(x)$,
and two noise levels, $\mu_0 > 0$ and $\mu_1 > 0$, define
\beqa
     \lefteqn{ \mse(p_1, p_0, \mu_1, \mu_0,z) } \nonumber \\
     &=&  \int_{x \in \Xset}
            |x-\xscaMmse(z \MID p_1,\mu_1)|^2 \,
            p_{x | z}(x \mid z \MID p_0,\mu_0) \, dx, \qquad
        \label{eq:mseScalar}
\eeqa
which is the MSE
in estimating the scalar $x$ from the variable $z$ in (\ref{eq:zmse})
when $x$ has a true distribution
$x \sim p_0(x)$ and the noise level is $\mu=\mu_0$,
but the estimator assumes a distribution
$x \sim p_1(x)$ and noise level $\mu=\mu_1$.

\medskip

\emph{Replica Symmetric Postulated MMSE Decoupling Property~\cite{GuoV:05}:}
Consider the estimation
problem in Section~\ref{sec:model}.
Let $\xbfpmmse(\ybf)$ be the PMMSE estimator based on a postulated
prior $\ppost$ and postulated noise level $\sigpost^2$.
For each $n$, let $j = j(n)$ be some deterministic component index with
$j(n) \in \{1,\ldots,n\}$.
Then under replica symmetry,
there exist \emph{effective noise levels}
$\sigEff^2$ and $\sigPEff^2$ such that:
\begin{itemize}
\item[(a)]   As $n \arr \infty$, the random vectors
$(x_j,s_j,\xpmmse_j)$ converge in distribution to the
random vector $(x,s,\xhat)$ consistent with the block diagram in Fig.~\ref{fig:scalarEquiv}.
Here $x$, $s$, and $v$ are independent with
$x \sim p_0(x)$, $s \sim p_S(s)$, $v \sim {\cal N}(0,1)$,
and
\begin{subequations}
\beqa
    \xhat &=& \xscaMmse(z \MID \ppost,\mu_p), \label{eq:xhatSE} \\
    z &=& x + \sqrt{\mu}v \label{eq:zSE},
\eeqa
where $\mu = \sigEff^2/s$ and $\mu_p = \sigPEff^2/s$.
\end{subequations}
\item[(b)] The effective noise levels satisfy the equations
\begin{subequations} \label{eq:repMmseFix}
\beqa
    \sigEff^2 &=& \sigtrue^2
        + \beta \,
        \Exp\left[ s \, \mse(\ppost,\ptrue,\mu_p,\mu,z) \right], \label{eq:sigEffSol} \\
    \sigPEff^2 &=& \sigpost^2 \nonumber \\
    & & + \beta \, \Exp\left[ s \, \mse(\ppost,\ppost,\mu_p,\mu_p,z)\right],
    \quad
            \label{eq:sigPEffSol}
\eeqa
\end{subequations}
where the expectations are taken over $s \sim p_S(s)$
and $z$ generated by (\ref{eq:zSE}).
\end{itemize}

\medskip

\begin{figure}
 \begin{center}
  \epsfig{figure=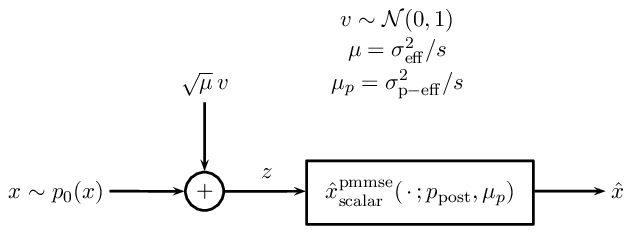,width=3.3in}
 \end{center}
\caption{Equivalent scalar model for the estimator behavior predicted
by the replica symmetric postulated MMSE decoupling property.}
\label{fig:scalarEquiv}
\end{figure}

This result asserts that the asymptotic
behavior of the joint estimation of the $n$-dimensional vector $\xbf$
can be described by $n$ equivalent scalar estimators.
In the scalar estimation problem, a component $x \sim \ptrue(x)$ is
corrupted by additive Gaussian noise yielding a
noisy measurement $z$.
The additive noise variance is $\mu = \sigEff^2/s$,
which is the effective noise divided by the scale factor $s$.
The estimate of that component is then described
by the (generally nonlinear) scalar estimator $\xscaMmse(z \MID \ppost,\mu_p)$.

The effective noise levels $\sigEff^2$ and $\sigPEff^2$
are described by the solutions to fixed-point equations (\ref{eq:repMmseFix}).  Note that $\sigEff^2$ and $\sigPEff^2$
appear implicitly on the left- and right-hand sides of these equations
via the terms $\mu$ and $\mu_p$.
In general, there is no closed form solution to these equations.
However, the expectations can be evaluated via
(one-dimensional)
numerical integration.

It is important to point out that there may, in general,
be multiple solutions to the fixed-point
equations (\ref{eq:repMmseFix}).
In this case, it turns out that the true solution is
the minimizer of a certain Gibbs' function described in~\cite{GuoV:05}.

\subsection{Effective Noise and Multiuser Efficiency}
To understand the significance of the effective noise level $\sigEff^2$,
it is useful to consider the following estimation problem with side information.
Suppose that when estimating the component $x_j$ an estimator is given
as side information the values of all the other components $\{x_\ell,\, \ell \neq j\}$.
Then, this hypothetical
estimator with side information can ``subtract out" the
effect of all the known components and compute
\[
    z_j = \frac{1}{\|\abf_j\|^2\sqrt{s_j}}\abf_j'\left(\ybf -
        \sum_{\ell \neq j} \sqrt{s_\ell} \, \abf_\ell \, x_\ell\right),
\]
where $\abf_\ell$ is the $\ell$th column of the measurement matrix $\Abf$.
It is easily checked that
\beqa
    z_j &=& \frac{1}{\|\abf_j\|^2\sqrt{s_j}}\abf_j'\left(\sqrt{s_j}\,\abf_j\,x_j
        + \wbf\right) \nonumber \\
        &=& x_j + \sqrt{\mu_0} \, v_j,
        \label{eq:zside}
\eeqa
where
\[
    v_j = \frac{1}{\sigtrue\|\abf_j\|^2}\abf_j'\wbf, \ \ \
    \mu_0 = \frac{\sigtrue^2}{s_j}.
\]
Thus, (\ref{eq:zside}) shows that with side information,
estimation of $x_j$ reduces to a scalar estimation problem where $x_j$
is corrupted by additive noise $\sqrt{\mu_0} \, v_j$.
Since $\wbf$ is Gaussian with mean zero and per-component variance $\sigtrue^2$,
$v_j$ is Gaussian with mean zero and variance $1/\|\abf_j\|^2$.
Also, since $\abf_j$ is an $m$-dimensional vector whose components are i.i.d.\
with variance $1/m$, $\|\abf_j\|^2 \arr 1$ as $m \arr \infty$.
Therefore, for large $m$, $v_j$ will approach $v_j \sim {\cal N}(0,1)$.

Comparing (\ref{eq:zside}) with (\ref{eq:zSE}), we see that the
equivalent scalar model predicted by the RS PMMSE decoupling property (\ref{eq:zSE})
is identical to the estimation with perfect side information (\ref{eq:zside}),
except that the noise level is increased by a factor
\beq \label{eq:muEff}
    1/\eta = \mu / \mu_0 = \sigEff^2 / \sigtrue^2.
\eeq
In multiuser detection, the factor $\eta$ is called the \emph{multiuser
efficiency}~\cite{Verdu:86,Verdu:98}.

The multiuser efficiency can be interpreted as degradation
in the effective signal-to-noise ratio (SNR):
With perfect side-information, an estimator
using $z_j$ in (\ref{eq:zside}) can estimate $x_j$ with
an effective SNR of
\beq \label{eq:snrNoMai}
    \SNR_0(s) = \frac{1}{\mu_0}\Exp|x_j|^2
    =  \frac{s}{\sigtrue^2}\Exp|x_j|^2.
\eeq
In CDMA multiuser detection, the factor $\SNR_0(s)$ is called
the post-despreading SNR with no multiple access interference.
The RS PMMSE decoupling property shows that without side information,
the effective SNR is given by
\beq \label{eq:snrMai}
    \SNR(s) = \frac{1}{\mu}\Exp|x_j|^2
    =  \frac{s}{\sigEff^2}\Exp|x_j|^2.
\eeq
Therefore, the multiuser efficiency $\eta$ in (\ref{eq:muEff})
is the ratio of the effective SNR with and without perfect
side information.

\section{Analysis of Postulated MAP Estimators via Hardening}
\label{sec:replicaMAP}
The main result of the paper is developed in this section.

\subsection{Postulated MAP Estimators}

Let $\Xset \subseteq \R$ be some (measurable) set and
consider an estimator of the form
\beq \label{eq:xhatMap}
    \xbfmap(\ybf) = \argmin_{\xbf \in \Xset^n}
        \frac{1}{2\gamma}\|\ybf - \Abf\Sbf^{1/2}\xbf\|^2_2 +
        \sum_{j=1}^n f(x_j),
\eeq
where $\gamma > 0$ is an algorithm parameter and
$f : \Xset \arr \R$ is some scalar-valued, nonnegative cost function.
We will assume that the objective function in (\ref{eq:xhatMap})
has a unique essential minimizer for almost all $\ybf$.

The estimator (\ref{eq:xhatMap}) can be interpreted as a MAP estimator.
To see this, suppose that for $u$ sufficiently large,
\beq \label{eq:puBnd}
    \int_{\xbf \in \Xset^n} e^{-uf(\xbf)} \, d\xbf < \infty,
\eeq
where we have extended the notation $f(\,\cdot\,)$ to vector arguments such that
\beq \label{eq:fsum}
    f(\xbf) = \sum_{j=1}^n f(x_j).
\eeq
When (\ref{eq:puBnd}) is satisfied,
we can define a prior probability distribution depending on $u$:
\beq \label{eq:pu}
    p_u(\xbf) = \left[\int_{\xbf \in \Xset^n}
    \exp(-uf(\xbf)) \, d\xbf\right]^{-1}\exp(-uf(\xbf)).
\eeq
Also, let
\beq \label{eq:sigu}
    \sigma_u^2 = \gamma / u.
\eeq
Substituting (\ref{eq:pu}) and (\ref{eq:sigu})
into (\ref{eq:pxy}), we see that
\beqa
   \lefteqn{p_{\xbf \mid \ybf}(\xbf \mid \ybf \MID p_u, \sigma_u^2)
   }\nonumber \\
     &=& C_u \exp\left[-u\left(\frac{1}{2\gamma}\|\ybf-\Abf\Sbf^{1/2}\xbf\|^2
    + f(\xbf) \right)\right] \label{eq:pxyu}
\eeqa
for some constant $C_u$ that does not depend on $\xbf$.
(The scaling of the noise variance along with $p_u$ enables the factorization
in the exponent of (\ref{eq:pxyu}).)
Comparing to (\ref{eq:xhatMap}), we see that
\[
    \xbfmap(\ybf) = \argmax_{\xbf \in \Xset^n}
    p_{\xbf \mid \ybf}(\xbf \mid \ybf \MID p_u, \sigma_u^2).
\]
Thus for all sufficiently large $u$,
we indeed have a MAP estimate---assuming the prior $p_u$ and
noise level $\sigma_u^2$.

\subsection{Decoupling under Replica Symmetric Assumption}
To analyze the postulated MAP (PMAP) estimator, we consider a sequence of
postulated MMSE estimators indexed by $u$.
For each $u$, let
\beq \label{eq:xhatu}
    \xbfhat^u(\ybf) = \Exp\left( \xbf \mid \ybf \MID p_u, \sigma_u^2 \right),
\eeq
which is the MMSE estimator of $\xbf$
under the postulated prior $p_u$ in
(\ref{eq:pu}) and noise level $\sigma^2_u$  in (\ref{eq:sigu}).
Using a standard large deviations argument, one can show that
under suitable conditions
\[
    \lim_{u \arr \infty} \xbfhat^u(\ybf) = \xbfmap(\ybf)
\]
for all $\ybf$.
A formal proof is given in Appendix~\ref{sec:xhatConv}
(see Lemma~\ref{lem:xhatuLim}).
Under the assumption that the behaviors of the postulated
MMSE estimators are described by the RS PMMSE decoupling property,
we can then extrapolate the behavior of the postulated MAP estimator.
This will yield our main result.

In statistical physics the parameter $u$ has the interpretation of
inverse temperature (see a general discussion in~\cite{Merhav:09arXiv}).
Thus, the limit as $u \arr \infty$ can be interpreted as a cooling or
``hardening" of the system.

In preparation for the main result,
define the scalar MAP estimator
\beq \label{eq:xmapScalar}
    \xscaMap(z \MID \lambda) = \argmin_{x \in \Xset} F(x,z,\lambda)
\eeq
where
\beq \label{eq:Fdef}
    F(x,z,\lambda) = \frac{1}{2\lambda} |z-x|^2 + f(x).
\eeq
The estimator (\ref{eq:xmapScalar}) plays a similar role as
the scalar MMSE estimator (\ref{eq:xmmseScalar}).

The main result
pertains to the estimator (\ref{eq:xhatMap}) applied
to the sequence of estimation problems defined in Section~\ref{sec:model}.
Our assumptions are as follows:
\begin{assumption} \label{as:replicaMMSE}
For all $u > 0$ sufficiently large, assume that the postulated MMSE estimator
\eqref{eq:xhatMmse}
with the postulated
prior $p_u$ in (\ref{eq:pu}) and postulated noise level $\sigma^2_u$ in (\ref{eq:sigu})
satisfy the RS PMMSE decoupling property in Section~\ref{sec:pmmse-decoupling}.
\end{assumption}
\begin{assumption} \label{as:limitExistence}
Let $\sigEff^2(u)$ and $\sigPEff^2(u)$
be the effective noise levels when using the postulated prior $p_u$
and noise level $\sigma^2_u$.  Assume the
following limits exist:
\begin{subequations}
\beqan
    \sigEffMap^2 &=& \lim_{u \arr \infty} \sigEff^2(u), \\
    \gamma_p &=& \lim_{u \arr \infty} u\sigPEff^2(u).
\eeqan
\end{subequations}
\end{assumption}
\begin{assumption} \label{as:limitExchange}
Suppose for each $n$, $\xhat^u_j(n)$ is the
MMSE estimate of the component $x_j$ for some index
$j \in \{1,\ldots,n\}$ based on the postulated prior $p_u$ and
postulated noise level $\sigma^2_u$.
Then, assume that limits can be interchanged to give the following equality:
\[
    \lim_{u \arr \infty} \lim_{n \arr \infty} \xhat^u_j(n)
    =  \lim_{n \arr \infty} \lim_{u \arr \infty} \xhat^u_j(n),
\]
where the limits are in distribution.
\end{assumption}
\begin{assumption} \label{as:uniqueness}
For every $n$, $\Abf$, and $\Sbf$, assume that for
almost all $\ybf$, the minimization in (\ref{eq:xhatMap})
achieves a unique essential minimum.  Here, essential should be understood in the standard measure-theoretic sense in that
the minimum and essential infimum agree.
\end{assumption}
\begin{assumption} \label{as:growth}
Assume that $f(x)$ is nonnegative and satisfies
\[
    \lim_{|x| \arr \infty} \frac{f(x)}{\log|x|} = \infty,
\]
where the limit must hold over all sequences in $\Xset$ with
$|x| \arr \infty$.  If $\Xset$ is compact,
this limit is automatically satisfied (since there are no sequences
in $\Xset$ with $|x|\arr \infty$).
\end{assumption}
\begin{assumption} \label{as:limitingVariance}
For all $\lambda \in \R$ and almost all $z$,
the minimization in (\ref{eq:xmapScalar}) has a unique, essential
minimum.  Moreover,
for all $\lambda$ and almost all $z$,
there exists a $\sigma^2(z,\lambda)$ such that
\beq \label{eq:sigzlam}
    \lim_{x \arr \xhat}
        \frac{|x-\xhat|^2}{2(F(x,z,\lambda) - F(\xhat,z,\lambda))}
        = \sigma^2(z,\lambda),
\eeq
where $\xhat =  \xscaMap(z \MID \lambda)$.
\end{assumption}

Assumption~\ref{as:replicaMMSE} is simply stated to again point out
that we are assuming the validity of replica symmetry for the postulated
MMSE estimates.   We make the additional
Assumptions~\ref{as:limitExistence} and~\ref{as:limitExchange},
which are also difficult to verify but similar in spirit.
Taken together, Assumptions~\ref{as:replicaMMSE}--\ref{as:limitExchange}
reflect the main limitations of the replica symmetric analysis and precisely
state the manner in which the analysis is non-rigorous.

Assumptions~\ref{as:uniqueness}--\ref{as:limitingVariance} are
technical conditions on the existence and uniqueness of the MAP estimate.
Assumption~\ref{as:uniqueness} will be true for any strictly convex regularization
$f(x_j)$, although it is difficult to verify in the non-convex case.
The other two assumptions,
Assumptions~\ref{as:growth} and~\ref{as:limitingVariance}, will be verified
for the problems of interest.
In fact, we will explicitly calculate $\sigma^2(z,\lambda)$.

We can now state our extension of the
RS PMMSE decoupling property.

\medskip

\emph{Replica Symmetric Postulated MAP Decoupling Property:}
Consider the estimation problem in Section~\ref{sec:model}.
Let $\xbfmap(\ybf)$ be the postulated MAP estimator (\ref{eq:xhatMap})
defined for some $f(x)$ and $\gamma > 0$
satisfying Assumptions~\ref{as:replicaMMSE}--\ref{as:limitingVariance}.
For each $n$, let $j = j(n)$ be some deterministic component index with
$j(n) \in \{1,\ldots,n\}$.
Then under replica symmetry (as part of Assumption~\ref{as:replicaMMSE}):
\begin{itemize}
\item[(a)]   As $n \arr \infty$, the random vectors
$(x_j,s_j,\xmap_j)$ converge in distribution to the
random vector $(x,s,\xhat)$
consistent with the block diagram in Fig.~\ref{fig:scalarEquivMap}
for the limiting \emph{effective noise levels} $\sigEff^2$ and $\gamma_p$
in Assumption~\ref{as:limitExistence}.
Here $x$, $s$, and $v$ are independent with
$x \sim p_0(x)$, $s \sim p_S(s)$, $v \sim {\cal N}(0,1)$,
and
\begin{subequations} \label{eq:xzMap}
\beqa
    \xhat &=& \xscaMap(z,\lambda_p), \label{eq:xhatSEMap} \\
    z &=& x + \sqrt{\mu}v, \label{eq:zSEMap}
\eeqa
where $\mu = \sigEffMap^2 / s$ and $\lambda_p = \gamma_p / s$.
\end{subequations}
\item[(b)] The limiting effective noise levels $\sigEffMap^2$
and $\gamma_p$ satisfy the equations
\begin{subequations} \label{eq:repMapFix}
\beqa
    \sigEffMap^2&=& \sigtrue^2
    + \beta\,\Exp\left[ s |x-\xhat|^2\right],
            \label{eq:sigEffSolMap} \\
    \gamma_p &=& \gamma
    + \beta\,\Exp\left[ s \sigma^2(z,\lambda_p)\right],
            \label{eq:sigPEffSolMap}
\eeqa
\end{subequations}
where the expectations are taken over $x \sim p_0(x)$, $s \sim p_S(s)$,
and $v \sim {\cal N}(0,1)$, with $\xhat$
and $z$ defined in (\ref{eq:xzMap}).
\end{itemize}
\begin{IEEEproof}
See Appendices~\ref{sec:proof-overview}--\ref{sec:fixed-point}.
\end{IEEEproof}

\begin{figure}
 \begin{center}
  \epsfig{figure=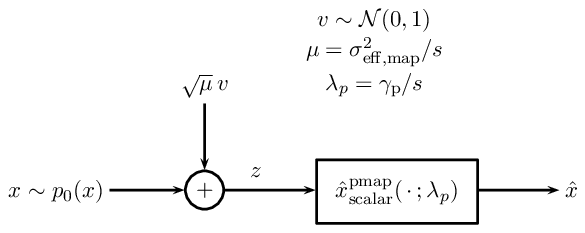,width=3.3in}
 \end{center}
\caption{Equivalent scalar model for the estimator behavior predicted
by the replica symmetric postulated MAP decoupling property.}
\label{fig:scalarEquivMap}
\end{figure}

Analogously to the RS PMMSE decoupling property,
the RS PMAP decoupling property
asserts that asymptotic behavior of the PMAP estimate of any single
component of $\xbf$ is described by a simple equivalent scalar estimator.
In the equivalent scalar model, the component of the true vector $\xbf$
is corrupted by Gaussian noise and the estimate
of that component is given by a scalar PMAP estimate of the
component from the noise-corrupted version.

\section{Analysis of Compressed Sensing}
\label{sec:compSensAnal}

Our results thus far hold for any separable distribution for $\xbf$
(see Section~\ref{sec:model})
and under mild conditions on the cost function $f$
(see especially Assumption~\ref{as:growth}, but other assumptions
also implicitly constrain $f$).
In this section, we provide additional details on replica analysis for
choices of $f$ that yield PMAP estimators relevant to compressed sensing.
Since the role of $f$ is to determine the estimator,
this is not the same as choosing sparse priors for $\xbf$.
Numerical evaluations of asymptotic performance with sparse priors
for $\xbf$ are given in Section~\ref{sec:sim}.

\subsection{Linear Estimation}
\label{sec:linAnal}

We first apply the RS PMAP decoupling property to the simple case of
linear estimation.
Linear estimators only use second-order statistics
and generally do not directly exploit sparsity or
other aspects of the distribution of the unknown vector $\xbf$.
Nonetheless, for sparse estimation problems, linear estimators
can be used as a first step in estimation, followed by thresholding
or other nonlinear operations~\cite{RauhutSV:08,FletcherRG:09-IT}.
It is therefore worthwhile to analyze the behavior of
linear estimators even in the context of sparse priors.

The asymptotic behavior of linear estimators with large random
measurement matrices is well known.
For example, using the Mar\v{c}enko-Pastur theorem~\cite{MarcenkoP:67},
Verd{\'u} and Shamai~\cite{VerduS:99}
characterized the behavior of linear estimators with large
i.i.d.\ matrices $\Abf$ and constant scale factors $\Sbf = I$.
Tse and Hanly~\cite{TseH:99} extended the analysis to general $\Sbf$.
{\GV}~\cite{GuoV:05} showed that both of these results can be recovered
as special cases of the general RS PMMSE decoupling property.
We show here that the RS PMAP decoupling property
can also recover these results.
Although the calculations are very similar to~\cite{GuoV:05},
and indeed we arrive at precisely the same results,
walking through the computations
will illustrate how the RS PMAP decoupling property is used.

To simplify the notation, suppose that the true prior on $\xbf$
is such that each component has zero mean and unit variance.
Choose the cost function
\[
    f(x) = \frac{1}{2}|x|^2,
\]
which corresponds to the negative log of a Gaussian prior also with zero
mean and unit variance.
With this cost function,
the PMAP estimator (\ref{eq:xhatMap}) reduces to the linear estimator
\beq \label{eq:xhatLin}
      \xbfmap(\ybf) = \Sbf^{1/2}\Abf'\left(\Abf\Sbf\Abf' + \gamma I\right)^{-1}
        \ybf.
\eeq
When $\gamma = \sigtrue^2$, the true noise variance,
the estimator (\ref{eq:xhatLin})
is the linear MMSE estimate.

Now, let us compute the effective noise levels from the
RS PMAP decoupling property.
First note that $F(x,z,\lambda)$ in (\ref{eq:Fdef}) is given by
\[
    F(x,z,\lambda) = \frac{1}{2\lambda}|z-x|^2 + \frac{1}{2}|x|^2,
\]
and therefore the scalar MAP estimator
in (\ref{eq:xmapScalar}) is given by
\beq \label{eq:xmapScalarLin}
    \xscaMap(z \MID \lambda) = \frac{1}{1+\lambda}z.
\eeq
A simple calculation also shows that $\sigma^2(z,\lambda)$ in
(\ref{eq:sigzlam}) is given by
\beq \label{eq:sigzlamLin}
    \sigma^2(z,\lambda) = \frac{\lambda}{1+\lambda}.
\eeq

As part (a) of the RS PMAP decoupling property, let
$\mu = \sigEffMap^2/s$ and $\lambda_p = \gamma_p/s$.
Observe that
\beqa
    \lefteqn{ \Exp\left[s\,|x-\xscaMap(z \MID \lambda_p)|^2\right] } \nonumber \\
    &\stackrel{(a)}{=}& \Exp\left[s\,\left|x-\frac{1}{1+\lambda_p}z\right|^2\right]
        \nonumber \\
    &\stackrel{(b)}{=}& \Exp\left[s\,\left|\frac{\lambda_p}{1+\lambda_p}x
        -\frac{\sqrt{\mu}}{1+\lambda_p}v\right|^2\right]
        \nonumber \\
    &\stackrel{(c)}{=}& \frac{s(\lambda_p^2 + \mu)}{(1+\lambda_p)^2},
        \label{eq:msepLin}
\eeqa
where (a) follows from (\ref{eq:xmapScalarLin});
(b) follows from (\ref{eq:zSEMap});
and (c) follows from the fact that $x$ and $v$ are uncorrelated with
zero mean and unit variance.
Substituting (\ref{eq:sigzlamLin}) and (\ref{eq:msepLin})
into the fixed-point equations (\ref{eq:repMapFix}),
we see that the limiting noise levels
$\sigEffMap^2$ and $\gamma_p$ must satisfy
\beqan
    \sigEffMap^2 &=& \sigtrue^2 +
        \beta \, \Exp\left[ \frac{s(\lambda_p^2 + \mu)}{(1+\lambda_p)^2}
            \right], \\
    \gamma_p &=& \gamma +
        \beta \, \Exp\left[ \frac{s\lambda_p}{1+\lambda_p} \right],
\eeqan
where the expectation is over $s \sim p_S(s)$.
In the case when $\gamma = \sigtrue^2$, it can be verified that
a solution to these fixed-point equations
is $\sigEffMap^2 = \gamma_p$, which results in $\mu = \lambda_p$ and
\beqa
    \sigEffMap^2 &=& \sigtrue^2 +
        \beta \, \Exp\left[ \frac{s\lambda_p}{1+\lambda_p} \right]
            \nonumber \\
    &=& \sigtrue^2 +
        \beta \, \Exp\left[ \frac{s\sigEffMap^2}{s+\sigEffMap^2} \right].
        \label{eq:sigEffLin}
\eeqa
The expression (\ref{eq:sigEffLin})
is precisely the Tse-Hanly formula~\cite{TseH:99}
for the effective interference.
Given a distribution on $s$,
this expression can be solved numerically for $\sigEffMap^2$.
In the special case of constant $s$, (\ref{eq:sigEffLin})
reduces to Verd{\'u} and Shamai's result in~\cite{VerduS:97}
and can be solved via a quadratic equation.

The RS PMAP decoupling property now states
that for any component index $j$,
the asymptotic joint distribution of $(x_j,s_j,\xhat_j)$
is described by $x_j$ corrupted by additive Gaussian
noise with variance $\sigEffMap^2/s$ followed by
a scalar linear estimator.

As described in~\cite{GuoV:05},
the above analysis can also be applied to other linear
estimators including the matched filter (where $\gamma \arr \infty$)
or the decorrelating receiver ($\gamma \arr 0$).

\subsection{Lasso Estimation}
\label{sec:lassoAnal}

We next consider lasso estimation, which is widely used
for estimation of sparse vectors.
The lasso estimate~\cite{Tibshirani:96}
(sometimes referred to as
basis pursuit denoising~\cite{ChenDS:99})
is given by
\beq \label{eq:xhatLasso}
    \xbflasso(\ybf) = \argmin_{\xbf \in \R^n}
        \frac{1}{2\gamma}\|\ybf - \Abf\Sbf^{1/2}\xbf\|^2_2 +
        \|\xbf\|_1,
\eeq
where $\gamma > 0$ is an algorithm parameter.
The estimator is essentially a least-squares estimator
with an additional $\|\xbf\|_1$ regularization term to
encourage sparsity in the solution.  The parameter $\gamma$
is selected to trade off the sparsity of the estimate with
the prediction error.
An appealing feature of lasso estimation is that the
minimization in (\ref{eq:xhatLasso}) is convex;
lasso thus enables computationally-tractable algorithms for finding
sparse estimates.

The lasso estimator (\ref{eq:xhatLasso}) is identical to the
PMAP estimator (\ref{eq:xhatMap}) with the cost function
\[
    f(x) = |x|.
\]
With this cost function, $F(x,z,\lambda)$ in (\ref{eq:Fdef}) is given by
\[
    F(x,z,\lambda) = \frac{1}{2\lambda}|z-x|^2 + |x|,
\]
and therefore the scalar MAP estimator
in (\ref{eq:xmapScalar}) is given by
\beq \label{eq:xmapScalarLasso}
    \xscaMap(z \MID \lambda) = \Tsoft_\lambda(z),
\eeq
where $\Tsoft_\lambda(z)$ is the soft thresholding operator
\beq \label{eq:Tsoft}
    \Tsoft_\lambda(z) = \left\{ \begin{array}{rl}
        z-\lambda, & \mbox{if $z > \lambda$}; \\
        0,         & \mbox{if $|z| \leq \lambda$}; \\
        z+\lambda, & \mbox{if $z < -\lambda$}.
        \end{array} \right.
\eeq

The RS PMAP decoupling property now states that there exists
effective noise levels $\sigEffMap^2$ and $\gamma_p$
such that for any component index $j$, the random vector
$(x_j,s_j,\xhat_j)$ converges in distribution to the
vector $(x,s,\xhat)$ where $x \sim \ptrue(x)$,
$s \sim p_S(s)$, and $\xhat$ is given by
\beq \label{eq:xtsoft}
    \xhat = \Tsoft_{\lambda_p}(z),
\qquad
    z = x + \sqrt{\mu}v,
\eeq
where $v \sim {\cal N}(0,1)$, $\lambda_p = \gamma_p/s$,
and $\mu = \sigEffMap^2/s$.
Hence, the asymptotic behavior of lasso has a remarkably
simple description:  the asymptotic distribution of the
lasso estimate $\xhat_j$ of the component $x_j$
is identical to $x_j$ being corrupted by Gaussian noise
and then soft-thresholded to yield the estimate $\xhat_j$.

This soft-threshold description has an appealing interpretation.
Consider the case when the measurement matrix $\Abf = I$.
In this case,
the lasso estimator (\ref{eq:xhatLasso}) reduces to
$n$ scalar estimates,
\beq \label{eq:xtsoftId}
    \xhat_j = \Tsoft_{\lambda}\left(x_j + \sqrt{\mu_0}v_j\right),
\qquad
    j = 1,\,2,\,\ldots,n,
\eeq
where $v_i \sim {\cal N}(0,1)$, $\lambda = \gamma/s$,
and $\mu_0 = \sigtrue^2/s$.
Comparing (\ref{eq:xtsoft}) and (\ref{eq:xtsoftId}),
we see that the asymptotic distribution of
$(x_j,s_j,\xhat_j)$ with large random $\Abf$ is identical
to the distribution in the trivial case where $\Abf = I$,
except that the noise levels
$\gamma$ and $\sigtrue^2$ are replaced by effective noise
levels $\gamma_p$ and $\sigEffMap^2$.

To calculate the effective noise levels,
one can perform a simple calculation to show that
$\sigma^2(z,\lambda)$ in
(\ref{eq:sigzlam}) is given by
\beq \label{eq:sigzLasso}
    \sigma^2(z,\lambda) = \left\{ \begin{array}{rl}
        \lambda, & \mbox{if $|z| > \lambda$}; \\
        0,       & \mbox{if $|z| \leq \lambda$}.
        \end{array} \right.
\eeq
Hence,
\beqa
    \Exp\left[s \sigma^2(z,\lambda_p)\right] &=&
    \Exp\left[s\lambda_p \Pr(|z| > \lambda_p) \right] \nonumber \\
    &=& \gamma_p\Pr(|z| > \gamma_p/s),
    \label{eq:sigSoft}
\eeqa
where we have used the fact that $\lambda_p = \gamma_p/s$.
Substituting (\ref{eq:xmapScalarLasso}) and (\ref{eq:sigSoft})
into (\ref{eq:repMapFix}),
we obtain the fixed-point equations
\begin{subequations} \label{eq:lassoFix}
\beqa
    \sigEffMap^2 &=& \sigtrue^2 +
        \beta \Exp\left[ s|x-\Tsoft_{\lambda_p}(z)|^2\right],
        \label{eq:lassoFixSig} \\
    \gamma_p &=& \gamma +
        \beta \gamma_p\Pr(|z| > \gamma_p/s), \label{eq:lassoFixGam}
\eeqa
\end{subequations}
where the expectations are taken with respect to $x \sim \ptrue(x)$,
$s \sim p_S(s)$, and $z$ in (\ref{eq:xtsoft}).
Again, while these fixed-point equations do not have a
closed-form solution, they can be relatively easily solved numerically
given distributions of $x$ and $s$.

\subsection{Zero Norm-Regularized Estimation}
\label{sec:zero}
Lasso can be regarded as a convex relaxation of zero norm-regularized estimator
\beq \label{eq:xhatZero}
    \xbfzero(\ybf) = \argmin_{\xbf \in \R^n}
        \frac{1}{2\gamma}\|\ybf - \Abf\Sbf^{1/2}\xbf\|^2_2 +
        \|\xbf\|_0,
\eeq
where $\|\xbf\|_0$ is the number of nonzero components
of $\xbf$.  For certain strictly sparse priors,
zero norm-regularized estimation may provide better performance
than lasso.  While \emph{computing} the zero norm-regularized estimate
is generally very difficult, we can use the replica
analysis to provide a simple prediction of its \emph{performance}.
This analysis can provide a bound on the performance achievable
by practical algorithms.

To apply the RS PMAP decoupling property to the zero norm-regularized estimator
(\ref{eq:xhatZero}), we observe that
the zero norm-regularized estimator is identical to the
PMAP estimator (\ref{eq:xhatMap}) with the cost function
\beq \label{eq:fzero}
    f(x) = \left\{ \begin{array}{rl}
        0, & \mbox{if $x = 0$}; \\
        1, & \mbox{if $x \neq 0$}. \end{array} \right.
\eeq
Technically, this cost function does not satisfy the
conditions of the RS PMAP decoupling property.  For one thing, without
bounding the range of $x$, the bound (\ref{eq:puBnd}) is not satisfied.
Also, the minimum of (\ref{eq:xmapScalar})
does not agree with the essential infimum.  To avoid these problems,
we can consider an approximation of (\ref{eq:fzero}),
\[
    f_{\delta,M}(x) =
    \left\{ \begin{array}{rl}
        0, & \mbox{if $|x| < \delta$}; \\
        1, & \mbox{if $|x| \in [\delta,M]$}, \end{array} \right.
\]
which is defined on the set $\Xset = \{x:|x| \leq M\}$.
We can then take the limits $\delta \arr 0$ and $M \arr \infty$.
For space considerations and to simplify the presentation,
we will just apply the decoupling property with $f(x)$
in (\ref{eq:fzero})
and omit the details of taking the appropriate limits.

With $f(x)$ given by (\ref{eq:fzero}),
the scalar MAP estimator in (\ref{eq:xmapScalar}) is given by
\beq \label{eq:xmapScalarZero}
    \xscaMap(z \MID \lambda) = \Thard_{t}(z),
\qquad
    t = \sqrt{2\lambda},
\eeq
where $\Thard_t$ is the hard thresholding operator,
\beq \label{eq:Thard}
    \Thard_t(z) = \left\{ \begin{array}{rl}
        z, & \mbox{if $|z| > t$}; \\
        0, & \mbox{if $|z| \leq t$}.
        \end{array} \right.
\eeq
Now, similar to the case of lasso estimation,
the RS PMAP decoupling property states that there exists
effective noise levels $\sigEffMap^2$ and $\gamma_p$
such that for any component index $j$, the random vector
$(x_j,s_j,\xhat_j)$ converges in distribution to the
vector $(x,s,\xhat)$ where $x \sim \ptrue(x)$,
$s \sim p_S(s)$, and $\xhat$ is given by
\beq \label{eq:xthard}
    \xhat = \Thard_t(z),
\qquad
    z = x + \sqrt{\mu}v,
\eeq
where $v \sim {\cal N}(0,1)$, $\lambda_p = \gamma_p/s$,
$\mu = \sigEffMap^2/s$, and
\beq \label{eq:tzero}
    t = \sqrt{2\lambda_p} = \sqrt{2\gamma_p/s}.
\eeq
Thus, the zero norm-regularized estimation
of a vector $\xbf$ is equivalent to
$n$ scalar components corrupted by some effective noise
level $\sigEffMap^2$ and hard-thresholded based on an effective
noise level $\gamma_p$.

The fixed-point equations for the effective noise levels
$\sigEffMap^2$ and $\gamma_p$
can be computed similarly to the case of lasso.
Specifically, one can verify that (\ref{eq:sigzLasso})
and (\ref{eq:sigSoft}) are both satisfied for
the hard thresholding operator as well.
Substituting (\ref{eq:sigSoft}) and (\ref{eq:xmapScalarZero})
into (\ref{eq:repMapFix}),
we obtain the fixed-point equations
\begin{subequations} \label{eq:zeroFix}
\beqa
    \sigEffMap^2 &=& \sigtrue^2 +
        \beta \Exp\left[ s|x-\Thard_{t}(z)|^2\right],
        \label{eq:zeroFixSig} \\
    \gamma_p &=& \gamma +
        \beta \gamma_p\Pr(|z| > t), \label{eq:zeroFixGam}
\eeqa
\end{subequations}
where the expectations are taken with respect to $x \sim \ptrue(x)$,
$s \sim p_S(s)$, $z$ in (\ref{eq:xthard}), and $t$
given by (\ref{eq:tzero}).
These fixed-point equations can be solved numerically.

\subsection{Optimal Regularization}
\label{sec:gamSel}

The lasso estimator (\ref{eq:xhatLasso}) and
zero norm-regularized estimator (\ref{eq:xhatZero})
require the setting of a regularization parameter $\gamma$.
Qualitatively, the parameter provides a mechanism to
trade off the sparsity level of the estimate
with the fitting error.
One of the benefits of the replica analysis is that it
provides a simple mechanism for optimizing the parameter
level given the problem statistics.

Consider first the lasso estimator (\ref{eq:xhatLasso})
with some $\beta > 0$ and distributions
$x \sim p_0(x)$ and $s \sim p_S(s)$.
Observe that there exists a solution to (\ref{eq:lassoFixGam})
with $\gamma > 0$ if and only if
\beq \label{eq:probzLim}
    \Pr\left(|z| > \gamma_p/s\right) < 1/\beta.
\eeq
This leads to a natural optimization:
we consider an optimization over two variables $\sigEffMap^2$ and
$\gamma_p$, where we minimize $\sigEffMap^2$
subject to (\ref{eq:lassoFixSig}) and (\ref{eq:probzLim}).

One simple procedure for performing this minimization is
as follows:  Start with
$t=0$ and some initial value of $\sigEffMap^2(0)$.  For any
iteration $t\geq 0$, we update $\sigEffMap^2(t)$ with the minimization
\beq \label{eq:sigEffMin}
    \sigEffMap^2(t+1) = \sigtrue^2 +
    \beta \min_{\gamma_p}
    \Exp\left[ s|x-\Tsoft_{\lambda_p}(z)|^2\right],
\eeq
where, on the right-hand side,
the expectation is taken over $x \sim \ptrue(x)$,
$s \sim p_S(s)$, $z$ in (\ref{eq:xtsoft}),
$\mu = \sigEffMap^2(t)/s$, and $\lambda_p = \gamma_p/s$.
The minimization in (\ref{eq:sigEffMin}) is over
$\gamma_p > 0$ subject to (\ref{eq:probzLim}).
One can show that with a sufficiently high initial condition,
the sequence $\sigEffMap^2(t)$ monotonically decreases to
a local minimum of the objective function.
Given the final value for $\gamma_p$, one can then
recover $\gamma$ from (\ref{eq:lassoFixGam}).
A similar procedure can be used for the zero norm-regularized
estimator.

\section{Numerical Simulations}
\label{sec:sim}

\subsection{Bernoulli--Gaussian Mixture Distribution}
As discussed above, the replica method is based on certain
unproven assumptions and even then the decoupling results
under replica symmetry are only asymptotic for the large dimension limit.
To validate the predictive power of
the RS PMAP decoupling property for finite dimensions, we first performed
numerical simulations where the components of $\xbf$
are a zero-mean Bernoulli--Gaussian process, or equivalently a
two-component, zero-mean Gaussian mixture where one component has zero variance.
Specifically,
\[
    x_j \sim \left\{ \begin{array}{rl}
        {\cal N}(0,1), & \mbox{with prob.\ $\rho$}; \\
        0,             & \mbox{with prob.\ $1-\rho$}, \\
    \end{array} \right.
\]
where $\rho$ represents a sparsity ratio.
In the experiments, $\rho = 0.1$.
This is one of many possible sparse priors.

We took the vector $\xbf$ to have $n=100$ i.i.d.\
components with this prior, and we varied $m$ for 10 different
values of $\beta = n/m$ from 0.5 to 3\@.
For the measurements (\ref{eq:yax}),
we took a measurement matrix $\Abf$ with i.i.d.\ Gaussian
components and a constant scale factor matrix $\Sbf = I$.
The noise level $\sigtrue^2$ was set so that $\SNR_0$ = 10~dB,
where $\SNR_0$ is the signal-to-noise ratio with perfect side
information defined in (\ref{eq:snrNoMai}).

We simulated various estimators
and compared their performances against the
asymptotic values predicted by the replica analysis.
For each value of $\beta$,
we performed 1000 Monte Carlo trials of each estimator.
For each trial, we measured the normalized squared error (SE) in dB
\[
    10 \log_{10}\left( \frac{
    \|\xbfhat - \xbf\|^2}{\|\xbf\|^2}
    \right),
\]
where $\xbfhat$ is the estimate of $\xbf$.
The results are shown in Fig.~\ref{fig:replicaSim},
with each set of 1000 trials represented by the median normalized SE in dB\@.

\begin{figure}
\begin{center}
  %\psfrag{beta}[][]{$\beta$}
  %\psfrag{m}[][]{$m$}
  % Good size for two column journal format:
  \epsfig{figure=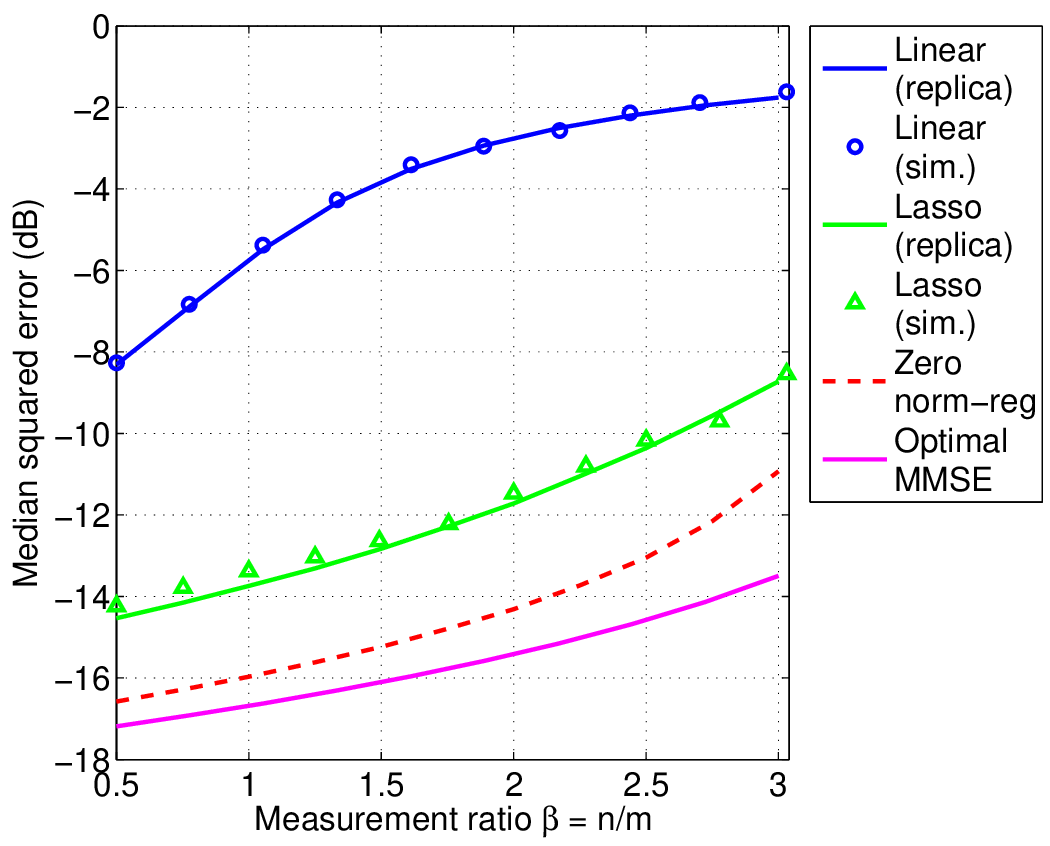,width=3.5in}
  % Good size for submission format:
  %\epsfig{figure=ml20.eps,width=6.6in}
\end{center}
\caption{MSE performance prediction with the RS PMAP decoupling property.
Plotted is the median normalized SE for
various sparse recovery algorithms:
linear MMSE estimation, lasso, zero norm-regularized estimation, and
optimal MMSE estimation.
Solid lines show the asymptotic predicted MSE from the
replica method.  For the linear and lasso estimators,
the circles and triangles show the actual median SE
over 1000 Monte Carlo simulations.
The unknown vector has i.i.d.\ Bernoulli--Gaussian components
with a 90\% probability of being zero.
The noise level is set so that $\captionSNR_0$ = 10~dB\@.
See text for details.}
\label{fig:replicaSim}
\end{figure}

The top curve shows the performance of the linear MMSE estimator
(\ref{eq:xhatLin}).  As discussed in Section~\ref{sec:linAnal},
the RS PMAP decoupling property applied to the case of a constant
scale matrix $\Sbf = I$
reduces to Verd{\'u} and Shamai's result in~\cite{VerduS:97}.
As can be seen in Fig.~\ref{fig:replicaSim}, the
result predicts the simulated performance of the linear
estimator extremely well.

The next curve shows the lasso estimator (\ref{eq:xhatLasso})
with the factor $\gamma$ selected to minimize the MSE
as described in Section~\ref{sec:gamSel}.
To compute the predicted value of the MSE from the RS PMAP decoupling property,
we numerically solve the fixed-point equations (\ref{eq:lassoFix})
to obtain the effective noise levels $\sigEffMap^2$
and $\gamma_p$.
We then use the scalar MAP model with the estimator
(\ref{eq:xmapScalarLasso}) to predict the MSE\@.
We see from Fig.~\ref{fig:replicaSim} that the predicted
MSE matches the median SE within 0.3~dB over a
range of $\beta$ values.
At the time of initial dissemination of this work~\cite{RanganFG:09arXiv},
precise prediction of lasso's performance given a specific
noise variance and prior was not achievable with any other method.
Now, as discussed in Section \ref{sec:bpConn},
such asymptotic performance predictions can also be proven rigorously through
connections with approximate belief propagation.

Fig.~\ref{fig:replicaSim} also shows the theoretical
minimum MSE (as computed with the RS PMMSE decoupling property)
and the theoretical MSE
from the zero norm-regularized estimator as computed in Section
\ref{sec:zero}.
For these two cases, the estimators cannot be simulated since
they involve NP-hard computations.
But we have depicted the curves to show that the replica method
can be used to calculate the gap between practical and impractical
algorithms.
Interestingly, we see that there is about a 2.0 to 2.5~dB gap between
lasso and zero norm-regularized estimation, and another
1 to 2~dB gap between zero norm-regularized estimation and optimal MMSE\@.

It is, of course, not surprising that zero norm-regularized estimation performs
better than lasso for the strictly sparse prior considered
in this simulation, and that optimal MMSE performs better yet.
However, what is valuable is that replica analysis
can quantify the precise performance differences.

In Fig.~\ref{fig:replicaSim}, we plotted the median SE
since there is actually considerable variation in the SE
over the random realizations of the problem parameters.
To illustrate the degree of variability, Fig.~\ref{fig:gmixCdf}
shows the CDF of the SE values over the 1000 Monte Carlo trials.
Each trial has different noise and measurement matrix realizations,
and both contribute to SE variations.
We see that the variation of the SE is especially large
at the smaller dimension $n=100$.
While the median value agrees well with the theoretical replica limit,
any particular instance of the problem can vary
considerably from that limit.
This is
a significant drawback of the replica method:
at lower dimensions,
the replica method may provide accurate predictions of the median
behavior, but it does not bound the variations from the median.

\begin{figure}
\begin{center}
  %\psfrag{beta}[][]{$\beta$}
  %\psfrag{m}[][]{$m$}
  % Good size for two column journal format:
  \epsfig{figure=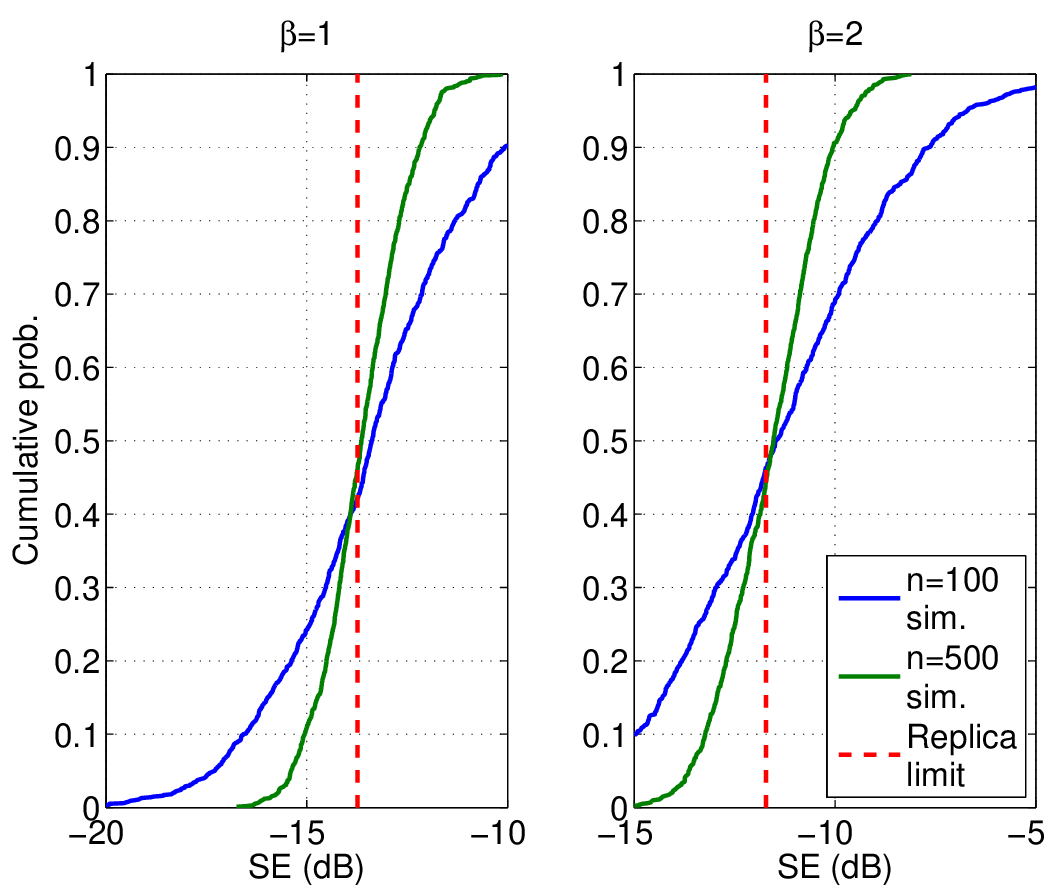,width=3.5in}
  % Good size for submission format:
  %\epsfig{figure=ml20.eps,width=6.6in}
\end{center}
\caption{Convergence to the asymptotic limit from the RS PMAP decoupling
property.
Plotted are the CDFs of the SE over 1000 Monte Carlo
trials of the lasso method for the Gaussian mixture
distribution. Details are in the text.  The CDF is shown
for dimensions $n=100$ and $n=500$ and
$\beta=1$ and $2$.  As vector dimension increases, the
performance begins to concentrate around the limit predicted
by the RS PMAP decoupling property.}
\label{fig:gmixCdf}
\end{figure}

As one might expect, at the higher dimension of $n=500$,
the level of variability is reduced and the observed SE
begins to concentrate around
the replica limit.  In his original paper~\cite{Tanaka:02},
Tanaka assumes that concentration of the SE will occur;
he calls this the \emph{self-averaging} assumption.
Fig.~\ref{fig:gmixCdf} provides some empirical evidence that
self-averaging does indeed occur.
However, even at $n=500$, the variation is not insignificant.
As a result, caution should be exercised in using the
replica predictions on particular low-dimensional instances.

\subsection{Discrete Distribution with Dynamic Range}

The RS PMAP decoupling property can also be used to study
the effects of dynamic range in power levels.
To validate the replica analysis with power variations,
we ran the following experiment:
the vector $\xbf$ was generated with i.i.d.\ components
\beq \label{eq:xpow}
    x_j = \sqrt{s_j} \, u_j,
\eeq
where $s_j$ is a random power level and $u_j$
is a discrete three-valued random variable with probability mass function
\beq \label{eq:bpskDist}
     u_j \sim  \left\{ \begin{array}{rl}
        1/\sqrt{\rho}, & \mbox{with prob } = \rho/2; \\
        -1/\sqrt{\rho}, & \mbox{with prob } = \rho/2; \\
        0, & \mbox{with prob } = 1-\rho.
        \end{array} \right.
\eeq
As before, the parameter $\rho$ represents the
sparsity ratio and we chose a value of $\rho =0.1$.
The measurements were generated by
\[
    \ybf = \Abf\xbf + \wbf = \Abf\Sbf^{1/2}\ubf + \wbf,
\]
where $\Abf$ is an i.i.d.\ Gaussian measurement matrix
and $\wbf$ is Gaussian noise.  As in the previous
section, the post-despreading SNR with side-information
was normalized to 10~dB\@.

The factor $s_j$ in (\ref{eq:xpow})
accounts for power variations in $x_j$.
We considered two random distributions for $s_j$:
(a)  $s_j = 1$, so that the power level is constant; and
(b) $s_j$ is uniform (in dB scale) over a 10~dB range with
unit average power.

In case (b), when there is variation in the power levels,
we can analyze two different scenarios for the lasso estimator:
\begin{itemize}
\item \emph{Power variations unknown:}
If the power level $s_j$ in (\ref{eq:xpow}) is unknown
to the estimator, then we can apply the standard
lasso estimator:
\beq
    \xbfhat(\ybf) = \argmin_{\xbf \in \R^n}
        \frac{1}{2\gamma}\|\ybf - \Abf\xbf\|^2_2 + \|\xbf\|_1,
\eeq
which does not need knowledge of the power levels $s_j$.
To analyze the behavior of this estimator
with the replica method, we
simply incorporate variations of both $u_j$ and $s_j$
into the prior of $x_j$
and assume a constant scale factor $s$ in the replica
equations.
\item \emph{Power variations known:}  If the power levels $s_j$
are known, the estimator can compute
\beq
    \ubfhat(\ybf) = \argmin_{\ubf \in \R^n}
        \frac{1}{2\gamma}\|\ybf - \Abf\Sbf^{1/2}\ubf\|^2_2
        + \|\ubf\|_1
\eeq
and then take $\xbfhat = \Sbf^{1/2}\ubfhat$.
This can be analyzed with the replica method by
incorporating the distribution of $s_j$ into the scale factors.
\end{itemize}

Fig.~\ref{fig:bpskPowSim} shows the performance of the lasso
estimator for the different power range scenarios.
As before, for each $\beta$, the figure plots the
median SE over 1000 Monte Carlo simulation trials.
Fig.~\ref{fig:bpskPowSim} also shows the theoretical asymptotic
performance as predicted with the RS PMAP decoupling property.
Simulated values are based on a vector dimension of $n=100$ and
optimal selection of $\gamma$ as described in Section~\ref{sec:gamSel}.

\begin{figure}
\begin{center}
  %\psfrag{beta}[][]{$\beta$}
  %\psfrag{m}[][]{$m$}
  % Good size for two column journal format:
  \epsfig{figure=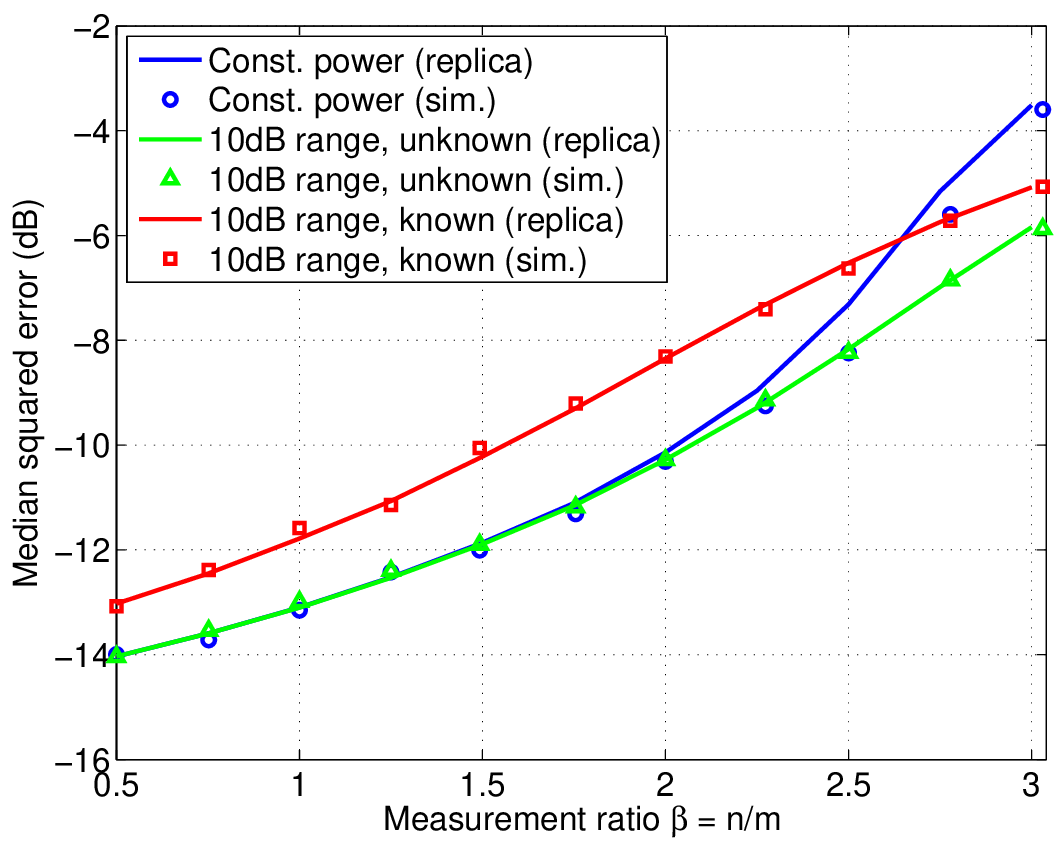,width=3.5in}
\end{center}
\caption{MSE performance prediction by the replica method
of the lasso estimator
with power variations in the components.
Plotted is the median SE of the lasso method in estimating
a discrete-valued distribution.
Three scenarios are considered:  (a) all components have the same power;
(b) the components have a 10~dB range in power that is unknown
to the estimator; and (c) the power range is known to the estimator
and incorporated into the measurement matrix.  Solid lines
represent the asymptotic prediction from the
RS PMAP decoupling property,
and the circles, triangles, and squares show the
median SE over 1000 Monte Carlo simulation.
See text for details.}
\label{fig:bpskPowSim}
\end{figure}

We see that in all three cases (constant power and power variations
unknown and known to the estimator), the replica prediction is
in excellent agreement with the simulated performance.
With one exception, the replica method matches the simulated performance
within 0.2~dB\@.   The one exception is for $\beta = 2.5$ with
constant power, where the replica method underpredicts the median
SE by about 1~dB\@.  A simulation at a higher dimension of $n=500$
(not shown here) reduced this discrepancy to 0.2~dB, suggesting that
the replica method is still asymptotically correct.

We can also observe two interesting
phenomena in Fig.~\ref{fig:bpskPowSim}.
First, the lasso method's performance
with constant power is almost identical to the performance
with unknown power variations for values of $\beta < 2$.
However, at higher values of $\beta$, the power variations
actually \emph{improve} the performance of the lasso method,
even though the average power is the same in both cases.
Wainwright's analysis~\cite{Wainwright:09-lasso}
demonstrated the significance of the minimum component
power in dictating lasso's performance.
The above simulation and the corresponding replica predictions
suggest that dynamic range may also play a role in the
performance of lasso.  That increased dynamic range can
improve the performance of sparse estimation has been
observed for other estimators~\cite{WipfR:06,FletcherRG:09arXiv}.

A second phenomena we see in Fig.~\ref{fig:bpskPowSim}
is that knowing the power variations
and incorporating them into the measurement matrix
can actually degrade the performance of lasso.  Indeed,
knowing the power variations appears to result in a 1 to 2~dB loss
in MSE performance.

Of course, one cannot conclude from this one simulation that
these effects of dynamic range hold more generally.
The study of the effect of dynamic range is interesting and
beyond the scope of this work.
The point is that the replica method provides a simple
analytic method for quantifying the effect of dynamic range
that appears to match actual performance well.

\subsection{Support Recovery with Thresholding}

In estimating vectors with strictly sparse priors,
one important problem is to detect the \emph{locations} of the
nonzero components in the vector $\xbf$.
This problem, sometimes called \emph{support recovery},
arises for example in subset selection in linear regression~\cite{Miller:02},
where finding the support of the vector $\xbf$
corresponds to determining a subset of features with
strong linear influence on some observed data $\ybf$.
Several works have attempted to find conditions under which
the support of a sparse vector $\xbf$ can be fully
detected~\cite{Wainwright:09-lasso,Wainwright:09-ml,FletcherRG:09-IT}
or partially detected~\cite{AkcakayaT:10,Reeves:08,AeronSZ:10}.
Unfortunately, with the exception of~\cite{Wainwright:09-lasso},
the only available results are bounds that are not tight.

One of the uses of RS PMAP decoupling property is
to \emph{exactly} predict the
fraction of support that can be detected correctly.
To see how to predict the support recovery performance,
observe that the decoupling property provides the
asymptotic joint distribution
for the vector $(x_j,s_j,\xhat_j)$, where $x_j$ is the
component of the unknown vector, $s_j$ is the corresponding
scale factor and $\xhat_j$ is the component estimate.
Now, in support recovery, we want to estimate $\theta_j$,
the indicator function that $x_j$ is nonzero
\[
    \theta_j = \left\{ \begin{array}{ll}
        1, & \mbox{if } x_j \neq 0;\\
        0, & \mbox{if } x_j \neq 0.
        \end{array} \right.
\]        
One natural estimate for $\theta_j$ is to compare the magnitude
of the component estimate $\xhat_j$ to some 
scale-dependent threshold $t(s_j)$,
\[
    \thetahat_j = \left\{ \begin{array}{ll}
        1, & \mbox{if } |\xhat_j| > t(s_j);\\
        0, & \mbox{if } |\xhat_j| \leq t(s_j).
        \end{array} \right.
\]
This idea of using thresholding for sparsity detection has been
proposed in~\cite{RauhutSV:08} and~\cite{SaligramaZ:11}.
Using the joint distribution $(x_j,s_j,\xhat_j)$, one can then
compute the probability of sparsity misdetection
\[
    p_{\rm err} = \Pr( \thetahat_j \neq \theta_j).
\]
The probability of error can be minimized over the threshold levels
$t(s)$.

To verify this calculation, we generated random vectors $\xbf$
with $n=100$ i.i.d.\ components
given by (\ref{eq:xpow}) and (\ref{eq:bpskDist}).
We used a constant power ($s_j=1$) and a sparsity fraction of
$\rho = 0.2$.  As before, the observations $\ybf$ were
generated with an i.i.d.\ Gaussian matrix with $\SNR_0$ = 10~dB\@.

\begin{figure}
\begin{center}
  %\psfrag{beta}[][]{$\beta$}
  %\psfrag{m}[][]{$m$}
  % Good size for two column journal format:
  \epsfig{figure=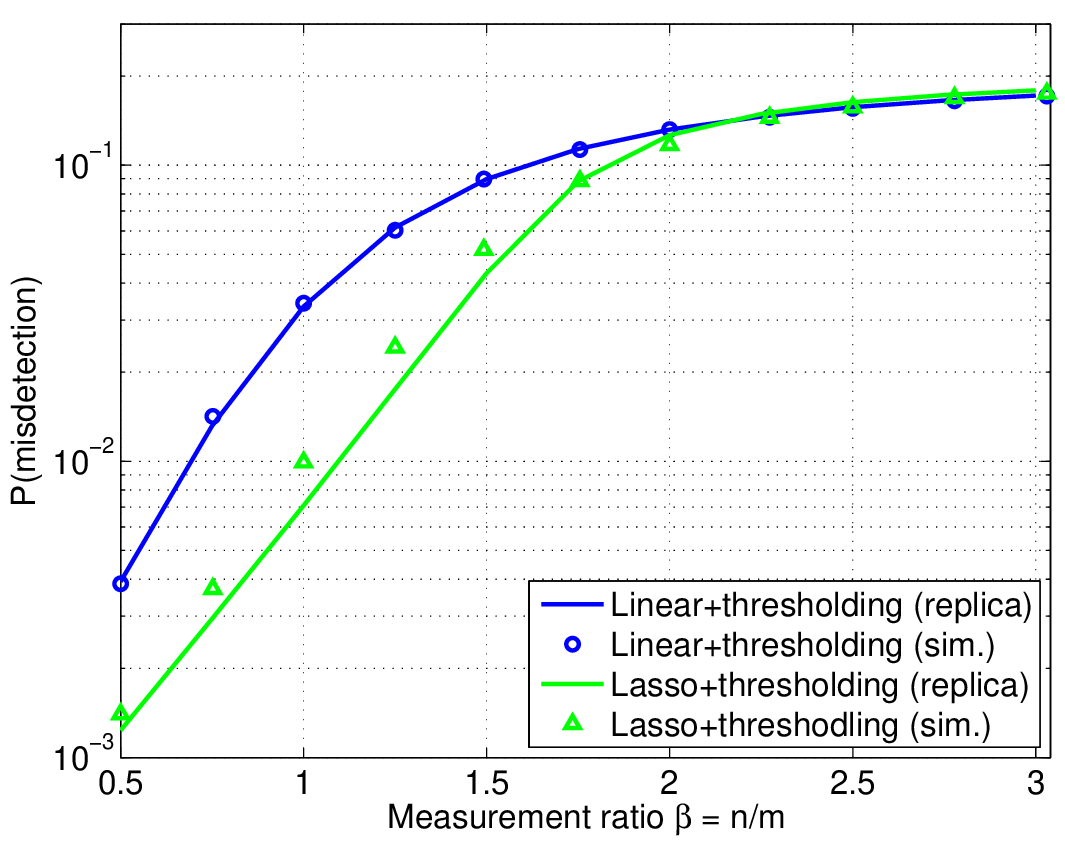,width=3.5in}
\end{center}
\caption{Support recovery performance prediction with
the replica method.  The solid lines show the theoretical probability
of error in sparsity misdetection using linear and lasso estimation
followed by optimal thresholding.
The circles and triangles are the corresponding
mean probabilities of misdetection
over 1000 Monte Carlo trials.}
\label{fig:bpskSparseDetect}
\end{figure}

Fig.~\ref{fig:bpskSparseDetect} compares the theoretical
probability of sparsity misdetection predicted by the replica method
against the actual probability of misdetection based on the average
of 1000 Monte Carlo trials.
We tested two algorithms:  linear MMSE estimation and lasso estimation.
For lasso, the regularization parameter was selected for minimum MMSE
as described in Section~\ref{sec:gamSel}.
The results show a good match.

\section{Conclusions and Future Work}
\label{sec:conclusions}

We have applied the replica method from statistical physics
for computing the asymptotic
performance of postulated MAP estimators of non-Gaussian vectors with
large random linear measurements, under a replica symmetric assumption.
The method can be readily applied to problems in compressed
sensing.  While the method is not theoretically rigorous,
simulations show an excellent ability to predict the performance
for a range of algorithms, performance metrics, and input
distributions.  Indeed, we believe that the replica method
provides the only method to date
for asymptotically-exact prediction of
performance of compressed sensing algorithms that can apply in
a large range of circumstances.

Moreover, we believe that
the availability of a simple scalar model that
exactly characterizes certain sparse estimators
opens up numerous avenues for analysis.
For one thing, it would be useful to see if
the replica analysis of lasso can be used to recover
the scaling laws of Wainwright~\cite{Wainwright:09-lasso}
and Donoho and Tanner~\cite{DonohoT:09}
for support recovery and to extend the latter to the noisy setting.
Also, the best known bounds for MSE performance in sparse estimation
are given by Haupt and Nowak~\cite{HauptN:06} and Cand{\`e}s and
Tao~\cite{CandesT:07}.  Since the replica analysis is asymptotically exact
(subject to various assumptions),
we may be able to obtain much tighter analytic expressions.
In a similar vein, several researchers have attempted
to find information-theoretic lower bounds with optimal
estimation~\cite{SarvothamBB:06-Allerton,FletcherRG:09-IT,Wainwright:09-ml}.
Using the replica analysis of optimal estimators, one may be able
to improve these scaling laws as well.

Finally, there is a well-understood connection between
statistical mechanics and belief propagation-based decoding
of error correcting codes~\cite{Sourlas:89,Montanari:00}.
These connections may suggest improved iterative algorithms
for sparse estimation as well.

\appendices

\section{Review of the Replica Method} \label{sec:replica}
We provide a brief summary of the replica method,
with a focus on some of the details of the
replica symmetric analysis of postulated MMSE estimation in~\cite{Tanaka:02,GuoV:05}.
This review will elucidate some of the key assumptions, notably the assumption of
replica symmetry.  General descriptions of the replica method can be found
in texts such as~\cite{Dotsenko:95,MezardParVir:01,MezardM:09,Talagrand:03}.

The replica method is based on evaluating variants of the
so-called \emph{asymptotic free energy}
\beq \label{eq:freeE}
    {\cal F} = -\lim_{n \arr \infty} \frac{1}{n} \Exp\left[ \log Z(\ybf,\Phi) \right],
\eeq
where  $Z(\ybf,\Phi)$ is the postulated partition function
\[
    Z(\ybf,\Phi) = \Exp\left[ \log p_{\ybf}(\ybf \mid \Phi \MID \ppost, \sigpost^2) \right]
\]
and the expectation in \eqref{eq:freeE} is with respect to the true distribution on $\ybf$.
For the replica PMMSE and PMAP analyses
in~\cite{Tanaka:02,GuoV:05}, various joint moments of the
variables $x_j$ and $\xhat_j$ are computed from certain variants of the free energy, and
the convergence of the joint distribution of $(x_j,\xhat_j)$ is then
analyzed based on these moments.

To evaluate the asymptotic free energy, the replica method uses the
identity that, for any random variable $Z$,
\[
    \Exp [\log Z] = \lim_{\nu \arr 0} \frac{\partial}{\partial \nu} \log \Exp\left[ Z^\nu \right].
\]
Therefore, the asymptotic free energy \eqref{eq:freeE} can be rewritten as
\beq \label{eq:freeEnu}
    {\cal F} = -\lim_{n \arr \infty} \frac{1}{n} \lim_{\nu \arr 0}
        \frac{\partial}{\partial \nu} \log \Exp\left[ Z^\nu(\ybf,\Phi) \right].
\eeq
The ``replica trick" involves evaluating the expectation $\Exp[Z^\nu(\ybf,\Phi)]$
for positive integer values of $\nu$ and then assuming an analytic continuation so that the
resulting expression is valid for real $\nu$ in the vicinity of zero.
For positive integer values of $\nu$, the quantity $Z^\nu(\ybf,\Phi)$ can be written as
\beq \label{eq:Znu}
    Z^\nu(\ybf,\Phi) = \Exp\left[ \prod_{a=1}^\nu p_{\ybf|\xbf}(\ybf \mid \xbf_a, \Phi \MID \ppost, \sigpost^2)
        \right],
\eeq
where the expectation is over independent copies of the vectors $\xbf_a$, $a=1,\ldots,\nu$,
with i.i.d.\ components $x_{aj} \sim \ppost(x_{aj})$.  The motivation for the replica trick is that
the quantity $Z^\nu(\ybf,\Phi)$ in \eqref{eq:Znu}
can be thought of as a partition function of a new system with $\nu$ ``replicated" copies
of the variables $\xbf_a$, $a=1,\ldots,\nu$.
The parameter $\nu$ is called the replica number.

The replicated system is relatively easy to analyze.
Specifically, to evaluate $\Exp[Z^\nu(\ybf,\Phi)]$, the replica analysis in~\cite{Tanaka:02,GuoV:05}
first assumes a \emph{self-averaging} property that essentially assumes that the variations
in $Z^\nu(\ybf,\Phi)$ due to randomness of the measurement matrix $\Phi$ vanish
in the limit as $n \arr \infty$.  Although a large number of statistical physics
quantities exhibit such self-averaging, the self-averaging of the relevant
quantities for the general PMMSE and PMAP analyses
has not been rigorously established.  Following~\cite{Tanaka:02,GuoV:05},
self-averaging in this work is thus simply assumed.

Under the self-averaging assumption, the expectation in \eqref{eq:Znu}
is evaluated in~\cite{GuoV:05}
by first conditioning on the $(\nu+1)$-by-$(\nu+1)$ correlation matrix $\Qbf = (1/n)\Xbf^T\Xbf$, where $\Xbf$ is the $n$-by-$(\nu+1)$ matrix
\[
    \Xbf = [\xbf \;\; \xbf_1 \;\; \ldots \;\; \xbf_\nu],
\]
with $\xbf$ having i.i.d.\ components according to the true distribution $x_j \sim p_0(x_j)$
and the vectors $\xbf_a$ being independent with i.i.d.\ components
following the postulated distribution $x_{aj} \sim \ppost(x_{aj})$.
The conditioning on $\Qbf$ reduces the expectation in \eqref{eq:Znu} to an integral of the form
\beqa
   \lefteqn{ \frac{1}{n}\Exp[Z^\nu(\ybf,\Phi)]} \nonumber \\
    &=& \frac{1}{n}\log\int \exp\left( \frac{n}{\beta} G^{(\nu)}(\Qbf) \right) \mu^{(\nu)}_n(d\Qbf) + O\left(\frac{1}{n}\right), \label{eq:EZnuInt}
\eeqa
where $G^{(\nu)}(\Qbf)$ is some function of the correlation matrix $\Qbf$ and
$\mu^{(\nu)}_n(\Qbf)$ is a probability measure on $\Qbf$.  It is then argued that
the measures $\mu^{\nu}_n(\Qbf)$ satisfy a large deviations property with some rate function
$I^{\nu}(\Qbf)$.
Then, using standard large deviations arguments as in~\cite{DemboZ:98}, the asymptotic
value of the expectation in \eqref{eq:EZnuInt} reduces to a maximization of the form
\beq
   \lim_{n \arr \infty} \Exp[Z^\nu(\ybf,\Phi)]
    = \sup_{\Qbf} \left[ \frac{1}{\beta}G^{(\nu)}(\Qbf) - I^{\nu}(\Qbf)\right],
    \label{eq:EZnuMax}
\eeq
where the supremum is over the set of covariance matrices $\Qbf$.
The correlation matrix $\Qbf$ plays a similar role as the so-called overlap
matrix in replica analyses of systems with discrete energy states~\cite{MezardM:09}.

The maximization in \eqref{eq:EZnuMax} over all covariance matrices is, in general,
difficult to perform.  The key replica symmetry (RS) assumption used in~\cite{Tanaka:02}
and~\cite{GuoV:05}, and hence implicitly used in this paper,
is that the maxima are achieved with matrices $\Qbf$ that are symmetric with respect
to permutations of the $\nu$ replica indices.  Under this symmetry
assumption, the space of covariance matrices is greatly reduced and the maxima \eqref{eq:EZnuMax}
can be explicitly evaluated.

The RS assumption is not always valid, even though the system itself is symmetric
across the replica indices.
For example, it is well-known
that even in the simple random energy model, the corresponding maximization
may not satisfy the RS assumption, particularly at low temperatures~\cite{MezardM:09};
see, also~\cite{NishimoriSher:01}.  More recently, it has been shown that replica symmetry
may also be broken when analyzing lattice precoding for the Gaussian broadcast
channel~\cite{ZaiMuMM:10arXiv}.

In absence of replica symmetry, one must search through a larger class of overlap or covariance matrices $\Qbf$.
One such hierarchy of classes of matrices that is often used is described by the
so-called $k$-step
replica symmetry breaking (RSB) matrices, a description of which can
be found in various texts~\cite{Dotsenko:95,MezardParVir:01,MezardM:09,Talagrand:03}.
In this regard, the analysis in this paper, which assumes replica symmetry, is thus only a 0-step RSB analysis or 0th-level prediction.

It is difficult to derive general tests for whether the RS assumption is
rigorously valid.
Tanaka's original work~\cite{Tanaka:02} derived an explicit condition for the validity
of the RS assumption based on the Almeida--Thouless (AT) 
test~\cite{AlmeidaThou:78} that
considers asymmetric perturbations around the RS saddle points of the
maximization \eqref{eq:EZnuMax}.
For the case of binary signals, the condition has a simple formula with the SNR
and measurement ratio $\beta$.
In~\cite{KabashimaWT:09arXiv}, an AT condition was also derived for RS analysis of 
$\ell_p$ reconstruction with Bernoulli--Gaussian priors.
Unfortunately, no equivalent condition has been derived for the general scenario considered in \GV's extension in~\cite{GuoV:05}.

In this work, we simply assume replica symmetry for the all values
of the scale factor $u > 0$.
Since $u$ is analogous to inverse temperature~\cite{Merhav:09arXiv}
and validity of the RS assumption is more problematic at low temperatures,
one must be cautious in interpreting our results.
As stated in Section~\ref{sec:intro}, where possible we have confirmed
the replica predictions by comparison to numerical experiments.  However, such experiments
are limited to computable estimators such as LASSO and linear estimators.  For other
estimators, such as the true MMSE or zero norm-regularized estimator, the RS assumption may very well
not hold.

\section{Proof Overview}
\label{sec:proof-overview}
Fix a deterministic sequence of indices $j = j(n)$
with $j(n) \in \{1,\ldots,n\}$.
For each $n$, define the random vector triples
\begin{subequations}
\beqa
    \thetau(n) &=& (x_j(n),s_j(n),\xhat^u_j(n)), \label{eq:thetau} \\
    \thetaMap(n) &=& (x_j(n),s_j(n),\xmap_j(n)),  \label{eq:thetaMap}
\eeqa
\end{subequations}
where $x_j(n)$, $\xhat^u_j(n)$, and $\xmap_j(n)$ are
the $j$th components of the random vectors $\xbf$,
$\xbfhat^u(\ybf)$, and $\xbfmap(\ybf)$,
and $s_j(n)$ is the $j$th diagonal entry of the matrix $\Sbf$.

For each $u$, we will use the notation
\beq \label{eq:xscau}
     \xscau(z \MID \lambda) = \xscaMmse(z \MID p_u,\lambda/u),
\eeq
where $p_u$ is defined in (\ref{eq:pu}) and
$\xscaMmse(z \MID \cdot,\cdot)$ is defined in (\ref{eq:xmmseScalar}).
Also, for every $\sigma$ and $\gamma > 0$
define the random vectors
\begin{subequations}
\beqa
    \thetaScau(\sigma^2,\gamma) &=& (x,s,\xscau(z \MID \gamma/s)),
        \label{eq:thetaScau}  \\
    \thetaScaMap(\sigma^2,\gamma) &=& (x,s,\xscaMap(z \MID \gamma/s)),
        \label{eq:thetaScaMap}
\eeqa
\end{subequations}
where $x$ and $s$ are independent with $x \sim p_0(x)$,
$s \sim p_S(s)$, and
\beq \label{eq:zxvsca}
    z = x + \frac{\sigma}{\sqrt{s}} v
\eeq
with $v \sim {\cal N}(0,1)$.

Now, to prove the RS PMAP decoupling property,
we need to show that (under the stated assumptions)
\beq \label{eq:thetaMapLim}
    \lim_{n \arr \infty} \thetaMap(n) = \thetaScaMap(\sigEffMap^2,\gamma_p),
\eeq
where the limit is in distribution and
the noise levels $\sigEffMap^2$ and
$\gamma_p$ satisfy part (b) of the claim.
This desired equivalence is depicted in the right column of
Fig.~\ref{fig:cs_commut}.

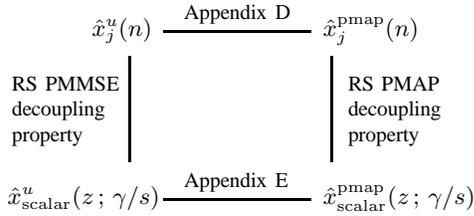
\begin{figure}
 \centering
  \setlength{\unitlength}{0.22in}
  \begin{picture}(6.0,5.0)(0.0,0.0)
   \put(5.2,1.0){\line(0,1){2.5}}
   \thicklines
   \put(-2.5,0.0){\small $\xscau(z \MID \gamma/s)$}
   \put(5.0,0.0){\small $\xscaMap(z \MID \gamma/s)$}
   \put(-0.5,4.0){\small $\xhat^u_j(n)$}
   \put(5.0,4.0){\small $\xmap_j(n)$}
   \put(1.2,0.1){\line(1,0){3.5}}
   \put(1.2,4.1){\line(1,0){3.5}}
   \put(0.4,1.0){\line(0,1){2.5}}
   \put(1.7,0.4){\footnotesize Appendix~\ref{sec:xhatScalarConv}}
   \put(1.7,4.4){\footnotesize Appendix~\ref{sec:xhatConv}}
   \put(-2.4,2.7){\footnotesize RS PMMSE}
   \put(-2.4,2.1){\footnotesize decoupling}
   \put(-2.4,1.5){\footnotesize property}
   \put(5.6,2.7){\footnotesize RS PMAP}
   \put(5.6,2.1){\footnotesize decoupling}
   \put(5.6,1.5){\footnotesize property}
  \end{picture}
  \caption{The RS PMAP decoupling property
    of this paper relates $\xmap_j(n)$ to
    $\xscaMap(z \MID \gamma/s)$ through an $n \rightarrow \infty$ limit.
    We establish the equivalence of its validity to the validity of the
    RS PMMSE decoupling property~\cite{GuoV:05} through two
    $u \rightarrow \infty$ limits:
    Appendix~\ref{sec:xhatConv} relates $\xhat^u_j(n)$ and $\xmap_j(n)$;
    Appendix~\ref{sec:xhatScalarConv} relates
    $\xscau(z \MID \gamma/s)$ and $\xscaMap(z \MID \gamma/s)$.}
 \label{fig:cs_commut}
\end{figure}

To show this limit we first observe that
under Assumption~\ref{as:replicaMMSE}, for $u$ sufficiently large,
the postulated prior distribution $p_u(x)$ in (\ref{eq:pu})
and noise level $\sigma^2_u$ in (\ref{eq:sigu}) are assumed to satisfy
the RS PMMSE decoupling property.
This
implies that
\beqa
    \lefteqn{ \lim_{n \arr \infty}
        (x_j(n),s_j(n),\xhat^u_j(n))} \nonumber \\
        &=&
        (x,s,\xscaMmse(z \MID p_u, \sigPEff^2(u)/s)), \label{eq:replicaLim}
\eeqa
where the limit is in distribution, $x \sim p_0(x)$, $s \sim p_S(s)$, and
\[
    z = x + \frac{\sigEff(u)}{\sqrt{s}}v,
\qquad
    v \sim {\cal N}(0,1).
\]
Using the notation above,
we can rewrite this limit as
\beqa
    \lim_{n \arr \infty} \thetau(n)
     &\stackrel{(a)}{=}& \lim_{n \arr \infty}
        (x_j(n),s_j(n),\xhat^u_j(n)) \nonumber \\
     &\stackrel{(b)}{=}&
        (x,s,\xscaMmse(z \MID p_u, \sigPEff^2(u)/s)) \nonumber \\
     &\stackrel{(c)}{=}&
        (x,s,\xscau(z \MID u\sigPEff^2(u)/s)) \nonumber \\
     &\stackrel{(d)}{=}& \thetaScau(\sigEff^2(u),u\sigPEff^2(u)),
     \label{eq:thetauLim}
\eeqa
where all the limits are in distribution and
(a) follows from the definition of $\thetau(n)$ in \eqref{eq:thetau};
(b) follows from \eqref{eq:replicaLim};
(c) follows from \eqref{eq:xscau};
and (d) follows from \eqref{eq:thetaScau}.
This equivalence is depicted in the left column of
Fig.~\ref{fig:cs_commut}.

The key part of the proof is to use
a large deviations argument to show that for almost all $\ybf$,
\[
    \lim_{u \arr \infty} \xbfhat^u(\ybf) = \xbfmap(\ybf).
\]
This limit in turn shows (see Lemma~\ref{lem:thetaLimn}
of Appendix~\ref{sec:xhatConv})
that for every $n$,
\beq \label{eq:thetaLimn}
    \lim_{u \arr \infty} \thetau(n) = \thetaMap(n)
\eeq
almost surely and in distribution.
A large deviation argument is also used to show that for
every $\lambda$ and almost all $z$,
\[
    \lim_{u \arr \infty} \xscau(z \MID \lambda) = \xscaMap(z \MID \lambda).
\]
Combining this with the limits in Assumption~\ref{as:limitExistence},
we will see (see Lemma~\ref{lem:thetaScaLim}
of Appendix~\ref{sec:xhatScalarConv}) that
\beqa
    \lefteqn{\lim_{u \arr \infty} \thetaScau(\sigEff^2(u),u\sigPEff^2(u)) }
        \nonumber\\
    &=& \thetaScaMap(\sigEffMap^2,\gamma_p)
     \label{eq:thetaLimSca}
\eeqa
almost surely and in distribution.

The equivalences (\ref{eq:thetaLimn}) and (\ref{eq:thetaLimSca}) are shown
as rows in Fig.~\ref{fig:cs_commut}.
As shown, they combine with the RS PMMSE decoupling property to prove
the RS PMAP decoupling property.
In equations instead of diagrammatic form,
the combination of limits is
\beqan
    \lim_{n \arr \infty} \thetaMap(n)
    &\stackrel{(a)}{=}& \lim_{n \arr \infty} \lim_{u \arr \infty}\thetau(n) \\
    &\stackrel{(b)}{=}& \lim_{u \arr \infty} \lim_{n \arr \infty}\thetau(n) \\
    &\stackrel{(c)}{=}& \lim_{u \arr \infty}
        \thetaScau(\sigEff^2(u),u\sigPEff^2(u)) \\
    &\stackrel{(d)}{=}& \thetaScaMap(\sigEffMap^2,\gamma_p)
\eeqan
where all the limits are in distribution and
(a) follows from (\ref{eq:thetaLimn});
(b) follows from Assumption~\ref{as:limitExchange};
(c) follows from (\ref{eq:thetauLim}); and
(d) follows from (\ref{eq:thetaLimSca}).
This proves (\ref{eq:thetaMapLim}) and part (a) of the claim.

Therefore,
to prove the claim we prove
the limit (\ref{eq:thetaLimn}) in Appendix~\ref{sec:xhatConv}
and
the limit (\ref{eq:thetaLimSca}) in Appendix~\ref{sec:xhatScalarConv}
and show that the limiting noise levels $\sigEffMap^2$ and $\gamma_p$
satisfy the fixed-point equations in part (b) of the claim
in Appendix~\ref{sec:fixed-point}.
Before these results are given, we review in Appendix~\ref{sec:deviations}
some requisite results from large deviations theory.

\section{Large Deviations Results}
\label{sec:deviations}

The above proof overview shows that the RS predictions for the postulated MAP estimate are
calculated by taking the limit as $u \arr \infty$ of the RS predictions of the postulated
MMSE estimates.
These limits are evaluated with large deviations theory and we begin,
in this appendix, by deriving some simple modifications of standard large deviations
results.  The main result we need is Laplace's principle as described in~\cite{DemboZ:98}:

\begin{lemma}[Laplace's Principle]  \label{lem:laplace}
Let $\varphi(\xbf)$ be any measurable function defined on some
measurable subset ${\cal D} \subseteq \R^n$ such that
\beq \label{eq:varphiBnd}
    \int_{\xbf \in {\cal D}} \exp(-\varphi(\xbf)) \, d\xbf < \infty.
\eeq
Then
\[
    \lim_{u \arr \infty} \frac{1}{u} \log \int_{\xbf \in {\cal D}}
         \exp(-u\varphi(\xbf)) \, d\xbf
        = -\essinf_{\xbf \in {\cal D}} \varphi(\xbf).
\]
\end{lemma}

Given $\varphi(\xbf)$ as in Lemma~\ref{lem:laplace},
define the probability distribution
\beq \label{eq:qu}
    q_u(\xbf) = \left[ \int_{\xbf\in {\cal D}}
        \exp(-u\varphi(\xbf)) \, d\xbf\right]^{-1} \exp(-u\varphi(\xbf)).
\eeq
We want to evaluate expectations of the form
\[
    \lim_{u \arr \infty} \int_{\xbf\in {\cal D}} g(u,\xbf)q_u(\xbf) \, d\xbf
\]
for some real-valued measurable function $g(u,\xbf)$.
The following lemma
shows that this integral is described by the behavior of $g(u,\xbf)$
in a neighborhood of the minimizer of $\varphi(\xbf)$.

\begin{lemma} \label{lem:largeDev}
Suppose that $\varphi(\xbf)$ and $g(u,\xbf)$ are real-valued
measurable functions such that:
\begin{itemize}
\item[(a)] The function $\varphi(\xbf)$ satisfies (\ref{eq:varphiBnd})
and has a unique essential minimizer $\xbfhat \in \R^n$
such that for every open neighborhood $U$ of $\xbfhat$,
\[
    \inf_{\xbf \not \in U} \varphi(\xbf) > \varphi(\xbfhat).
\]
\item[(b)] The function $g(u,\xbf)$ is positive and satisfies
\[
    \limsup_{u \arr \infty} \sup_{\xbf \not \in U}
        \frac{\log g(u,\xbf) }{u(\varphi(\xbf) - \varphi(\xbfhat))} \leq 0
\]
for every open neighborhood $U$ of $\xbfhat$.
\item[(c)] There exists a constant $g_\infty$ such that
for every $\epsilon > 0$, there exists a neighborhood
$U$ of $\xhat$ such that
\[
    \limsup_{u \arr \infty}
    \left| \int_U g(u,\xbf)q_u(\xbf) \, dx -g_\infty\right| \leq \epsilon.
\]
\end{itemize}
Then,
\[
    \lim_{u \arr \infty} \int g(u,\xbf)q_u(\xbf) \, d\xbf = g_\infty.
\]
\end{lemma}
\begin{IEEEproof} Due to item (c), we simply have to show
that for any open neighborhood $U$ of $\xbfhat$,
\[
    \limsup_{u \arr \infty} \int_{\xbf \in U^c}
     g(u,\xbf)q_u(\xbf) \, d\xbf = 0.
\]
To this end, let
\[
    Z(u) = \log \int_{\xbf \in U^c} g(u,\xbf)q_u(\xbf) \, d\xbf.
\]
It suffices to show that $Z(u) \rightarrow -\infty$ as $u \rightarrow \infty$.
Using the definition of $q_u(\xbf)$ in (\ref{eq:qu}),
it is easy to check that
\beq \label{eq:zdiff}
    Z(u) = Z_1(u) - Z_2(u),
\eeq
where
\beqan
    Z_1(u) &=& \log \int_{\xbf \in U^c}
        g(u,\xbf)\exp\left(-u(\varphi(\xbf)-\varphi(\xbfhat))\right) \, d\xbf, \\
    Z_2(u) &=& \log \int_{\xbf \in {\cal D} }
        \exp\left(-u(\varphi(\xbf)-\varphi(\xbfhat))\right) \, d\xbf.
\eeqan
Now, let
\[
    M = \essinf_{\xbf \in U^c} \varphi(\xbf) - \varphi(\xbfhat).
\]
By item (a), $M > 0$.
Therefore, we can find a $\delta > 0$ such that
\beq \label{eq:delmdef}
    -M(1-\delta) + 3\delta < 0.
\eeq
Now, from item (b), there exists a $u_0$ such that for all $u > u_0$,
\beqan
    Z_1(u) &\leq& \log \int_{\xbf \in U^c}
        \exp(-u(1-\delta)(\varphi(\xbf)-\varphi(\xbfhat))) \, d\xbf.
\eeqan
By Laplace's principle, we can find a $u_1$ such that
for all $u > u_1$,
\beqa
    Z_1(u) &\leq& u\left[\delta -
    \inf_{\xbf \in U^c}(1-\delta)
    (\varphi(\xbf) - \varphi(\xbfhat))\right]  \nonumber \\
     &=& u(-M(1-\delta) + \delta). \label{eq:zbnda}
\eeqa
Also, since $\xbfhat$ is an essential minimizer of $\varphi(\xbf)$,
\[
    \essinf_{\xbf \in {\cal D}} \varphi(\xbf) = \varphi(\xbfhat).
\]
Therefore, by Laplace's principle, there exists a $u_2$ such that
for $u > u_2$,
\beq \label{eq:zbndb}
    Z_2(u) \geq u\left[-\delta -
    \essinf_{\xbf \in {\cal D}}
    (\varphi(\xbf) - \varphi(\xbfhat))\right]  \\
     = -u\delta.
\eeq
Substituting (\ref{eq:zbnda}) and (\ref{eq:zbndb}) into
(\ref{eq:zdiff}) we see that for $u$ sufficiently large,
\[
    Z(u) \leq
        u(-M(1-\delta) + \delta) + u\delta < -u\delta,
\]
where the last inequality follows from (\ref{eq:delmdef}).
This shows $Z(u) \arr -\infty$ as $u \arr \infty$
and the proof is complete.
\end{IEEEproof}

One simple application of this lemma is as follows:

\begin{lemma} \label{lem:largeDevPt}
Let $\varphi(\xbf)$ and $h(\xbf)$ be real-valued measurable functions
satisfying the following:
\begin{itemize}
\item[(a)] The function $\varphi(\xbf)$ has a unique essential
minimizer $\xbfhat$
such that for every open neighborhood $U$ of $\xbfhat$,
\[
    \inf_{\xbf \not \in U} \varphi(\xbf) > \varphi(\xbfhat).
\]
\item[(b)] The function $h(\xbf)$ is continuous at $\xbfhat$.
\item[(c)] There exists a $c > 0$ and compact set $K$
such that for all $\xbf \not \in K$,
\beq \label{eq:varphilogh}
    \varphi(\xbf) \geq c\log|h(\xbf)|.
\eeq
\end{itemize}
Then,
\[
    \lim_{u \arr \infty} \int h(\xbf)q_u(\xbf) \, d\xbf = h(\xbfhat).
\]
\end{lemma}
\begin{IEEEproof}
We will apply Lemma~\ref{lem:largeDev} with
$g(u,\xbf) = |h(\xbf) - h(\xbfhat)|$ and
$g_\infty = 0$.
Item (a) of this lemma shows that $\varphi(\xbf)$
satisfies item (a) in Lemma~\ref{lem:largeDev}.

To verify that item (b) of Lemma~\ref{lem:largeDev} holds,
we first claim there exists a constant $M > 0$ such that for all
$\xbf$,
\beq \label{eq:varphihm}
    \varphi(\xbf)-\varphi(\xbfhat) \geq M\log|h(\xbf)-h(\xbfhat)|.
\eeq
We find a valid constant $M$ for three regions.
First, let $U$ be the set of $\xbf$ such that $|h(\xbf)-h(\xbfhat)| < 1$.
Since $h(\xbf)$ is continuous in $\xbf$,
$U$ is an open neighborhood of $\xbfhat$.
Also, for $\xbf \in U$, the right hand side of \eqref{eq:varphihm} is negative.
Since the left-hand side of \eqref{eq:varphihm} is positive,
the inequality will be satisfied in $U$ for any $M > 0$.

Next, consider the set $K_1 = K \backslash U$ where $K$ is the compact set
in item (c) of this lemma.  Since $K$ is compact and $h(\xbf)$ is continuous,
there exists a $c_1 > 0$ such that $\log|h(\xbf)-h(\xbfhat)| < c_1$ for all $\xbf \in K$.
Also, since $U$ is an open neighborhood of $\xbfhat$,
by item (a), there exists a $c_2 > 0$ such that $\varphi(\xbf)-\varphi(\xbfhat) \geq c_2$
for all $\xbf \not \in U$.  Hence, the inequality \eqref{eq:varphihm}
is satisfied with $M = c_2/c_1$ in the set $K_1$.

Finally, consider the set $K^c$. In this set,  \eqref{eq:varphilogh} is satisfied for some $c>0$.
Combining this inequality with the fact that $\varphi(\xbf)-\varphi(\xbfhat)\geq c_2$ for some $c_2>0$,
one can show that \eqref{eq:varphihm} also holds for some $M > 0$.  Hence,
for each of the regions $U$, $K \backslash U$ and $K^c$,
\eqref{eq:varphihm} is satisfied for some $M > 0$.  Taking the maximum of the three values of $M$,
one can assume \eqref{eq:varphihm} for all $\xbf$.

Applying \eqref{eq:varphihm}, we obtain
\[
    \frac{\log g(u,\xbf)}{\varphi(\xbf) - \varphi(\xbfhat)}
    = \frac{\log|h(\xbf) - h(\xbfhat)|}
    {\varphi(\xbf) - \varphi(\xbfhat)} \leq \frac{1}{M}.
\]
Hence, for any open neighborhood $U$ of $\xbfhat$,
\[
    \limsup_{u \arr \infty} \sup_{\xbf \not \in U}
        \frac{\log g(u,\xbf) }{u(\varphi(\xbf) - \varphi(\xbfhat))}
        \leq \lim_{u \arr \infty} \frac{1}{uM} = 0.
\]

Now let us verify that item (c) of Lemma~\ref{lem:largeDev} holds.
Let $\epsilon > 0$.
Since $h(\xbf)$ is continuous at $\xbf$, there exists an open neighborhood
$U$ of $\xbf$ such that $g(u,\xbf) < \epsilon$ for all $\xbf \in U$ and $u$.
This implies that for all $u$,
\[
    \int_U g(u,\xbf)q_u(\xbf) \, d\xbf
      < \epsilon \int_U q_u(\xbf) \, d\xbf
      \leq \epsilon,
\]
which shows that $g(u,\xbf)$ satisfies item (c) of Lemma~\ref{lem:largeDev}.
Thus
\beqan
    \lefteqn { \left| \int h(\xbf)q_u(\xbf) \, d\xbf - h(\xbfhat)\right|  } \\
    &=&   \left| \int (h(\xbf)-h(\xbfhat))q_u(\xbf) \, d\xbf \right|  \\
    &\leq&   \int |h(\xbf)-h(\xbfhat)|q_u(\xbf) \, d\xbf   \\
    &\leq&   \int g(u,\xbf)q_u(\xbf) \, d\xbf  \arr 0,
\eeqan
where the last limit is as $u \arr \infty$ and follows from
Lemma~\ref{lem:largeDev}.
\end{IEEEproof}

\section{Evaluation of  $\lim_{u \arr \infty} \xbfhat^u(\ybf)$}
\label{sec:xhatConv}

We can now apply Laplace's principle in the previous section
to prove (\ref{eq:thetaLimn}).  We begin by examining the pointwise
convergence of the PMMSE estimator $\xbfhat^u(\ybf)$.

\begin{lemma} \label{lem:xhatuLim}
For every $n$, $\Abf$, and $\Sbf$
 and almost all $\ybf$,
\[
    \lim_{u \arr \infty} \xbfhat^u(\ybf) = \xbfmap(\ybf),
\]
where $\xbfhat^u(\ybf)$ is the PMMSE estimator in (\ref{eq:xhatu}) and
$\xbfmap(\ybf)$ is the PMAP estimator in (\ref{eq:xhatMap}).
\end{lemma}
\begin{IEEEproof}
The lemma is a direct application of Lemma~\ref{lem:largeDevPt}.
Fix $n$, $\ybf$, $\Abf$, and $\Sbf$ and let
\beq \label{eq:varphiMap}
    \varphi(\xbf) = \frac{1}{2\lambda}\|\ybf-\Abf\Sbf^{1/2}\xbf\|^2 +
        f(\xbf).
\eeq
The definition of $\xbfmap(\ybf)$ in (\ref{eq:xhatMap}) shows that
\[
    \xbfmap(\ybf) = \argmin_{\xbf \in \Xset^n} \varphi(\xbf).
\]
Assumption~\ref{as:uniqueness} shows that this minimizer is unique
for almost all $\ybf$.
Also (\ref{eq:pxyu}) shows that
\beqan
    \lefteqn{
    p_{\xbf \mid \ybf}(\xbf \mid \ybf \MID p_u, \sigma_u^2) } \\
    &=& \left[ \int_{\xbf\in\Xset^n}
    \exp\left(-u\varphi(\xbf)\right) \, d\xbf\right]^{-1}
        \exp(-u\varphi(\xbf)) \\
    &=& q_u(\xbf),
\eeqan
where $q_u(\xbf)$ is given in (\ref{eq:qu}) with ${\cal D} = \Xset^n$.
Therefore, using (\ref{eq:xhatu}),
\beq \label{eq:xhatuInt}
    \xbfhat^u(\ybf) = \Exp\left( \xbf \mid \ybf \MID p_u, \sigma_u^2 \right)
    = \int_{\xbf \in \Xset^n} \xbf \, q_u(\xbf) \, d\xbf.
\eeq

Now, to prove the lemma, we need to show that
\[
    \lim_{n \arr \infty} \xhat^u_j(\ybf) = \xmap_j(\ybf)
\]
for every component $j = 1,\ldots,n$.
To this end, fix a component index $j$.
Using (\ref{eq:xhatuInt}), we can write
the $j$th component of $\xbfhat^u(\ybf)$ as
\[
    \xhat_j^u(\ybf) = \int_{\xbf \in \Xset^n} h(\xbf) q_u(\xbf) \, d\xbf,
\]
where $h(\xbf) = x_j$.  The function $h(\xbf)$ is continuous.
To verify item (c) of Lemma~\ref{lem:largeDevPt}, using Assumption~\ref{as:growth},
we first find a compact set $K$ such that $|\xbf| \not \in K$ implies that
\beq \label{eq:flogxBnd}
    f(x_j)  > c\log|x_j|.
\eeq
Then, for $\xbf \not \in K$,
\[
    \varphi(\xbf) \stackrel{(a)}{\geq} f(\xbf) \stackrel{(b)}{\geq} f(x_j)
        \stackrel{(c)}{\geq} c\log|x_j|,
\]
where (a) follows from the definition of $\varphi(\xbf)$ in \eqref{eq:varphiMap};
(b) follows from \eqref{eq:fsum} and the assumption that the cost functions $f(x_i)$
are non-negative; and (c) follows from \eqref{eq:flogxBnd}.
Therefore, item (c) of Lemma~\ref{lem:largeDevPt} follows since
$h(x_j) = x_j$. Thus, all the hypotheses of Lemma~\ref{lem:largeDevPt}
are satisfied and we have the limit
\[
    \lim_{u \arr \infty} \xhat_j^u(\ybf) = h(\xbfmap(\ybf)) = \xmap_j(\ybf).
\]
This proves the lemma.
\end{IEEEproof}

\begin{lemma} \label{lem:thetaLimn}
Consider the random vectors
$\thetau(n)$ and $\thetaMap(n)$ defined in (\ref{eq:thetau})
and (\ref{eq:thetaMap}), respectively.  Then, for all $n$,
\beq \label{eq:thetaLimnLem}
    \lim_{u \arr \infty} \thetau(n) = \thetaMap(n)
\eeq
almost surely and in distribution.
\end{lemma}
\begin{IEEEproof}  The vectors $\thetau(n)$
and $\thetaMap(n)$  are deterministic functions
of $\xbf(n)$, $\Abf(n)$, $\Sbf(n)$, and $\ybf$.
Lemma~\ref{lem:xhatuLim} shows that the limit
(\ref{eq:thetaLimnLem}) holds for any values
of $\xbf(n)$, $\Abf(n)$, and $\Sbf(n)$, and almost all $\ybf$.
Since $\ybf$ has a continuous probability distribution
(due to the additive noise $\wbf$ in (\ref{eq:yax})),
the set of values where this limit does not hold must
have probability zero. Thus, the limit (\ref{eq:thetaLimnLem})
holds almost surely, and therefore, also in distribution.
\end{IEEEproof}

\section{Evaluation of  $\lim_{u \arr \infty} \xscau(z \MID \lambda)$}
\label{sec:xhatScalarConv}
We first show the pointwise convergence of
the scalar MMSE estimator $\xscau(z \MID \lambda)$.

\begin{lemma} \label{lem:xhatScauLim}
Consider the scalar estimators
$\xscau(z \MID \lambda)$ defined in (\ref{eq:xscau})
and $\xscaMap(z \MID \lambda)$ defined in (\ref{eq:xmapScalar}).
For all $\lambda > 0$ and almost all $z$,
we have the deterministic limit
\[
    \lim_{u \arr \infty} \xscau(z \MID \lambda) = \xscaMap(z \MID \lambda).
\]
\end{lemma}
\begin{IEEEproof}
The proof is similar to that of Lemma~\ref{lem:xhatuLim}.
Fix $z$ and $\lambda$ and consider the conditional distribution
$p_{x|z}(x \mid z \MID p_u,\lambda/u)$.
Using (\ref{eq:pxz}) along with
the definition of $p_u(x)$ in (\ref{eq:pu}) and an argument
similar to the proof of Lemma~\ref{lem:xhatuLim},
it is easily checked that
\beq \label{eq:pxzu}
    p_{x|z}(x \mid z \MID p_u,\lambda/u) = q_u(x),
\eeq
where $q_u(x)$ is given by (\ref{eq:qu}) with ${\cal D} = \Xset$
and
\beq \label{eq:varphiSca}
    \varphi(x) = F(x,z,\lambda),
\eeq
where $F(x,z,\lambda)$ is defined in (\ref{eq:Fdef}).
Using (\ref{eq:xscau}) and (\ref{eq:xmmseScalar}),
\beqan
    \lefteqn{ \xscau(z \MID \lambda) = \xscaMmse(z \MID p_u,\lambda/u) }\\
    &=& \int_{x \in \Xset} x \, p_{x\mid z}(x \mid z \MID p_u,\lambda/u) \, dx \\
    &=& \int_{x \in \Xset} h(x) q_u(x) \, dx,
\eeqan
with $h(x) = x$.

We can now apply Lemma~\ref{lem:largeDevPt}.
The definition of $\xscaMap(z \MID \lambda)$ in
(\ref{eq:xmapScalar}) shows that
\beq \label{eq:xscaMapComp}
    \xscaMap(z \MID \lambda) = \argmin_{x \in \Xset} \varphi(x).
\eeq
Assumption~\ref{as:limitingVariance} shows that
for all $\lambda > 0$ and almost all $z$, this minimization is unique so
\[
    \varphi(x) > \varphi(\xscaMap(z \MID \lambda))
\]
for all $x \neq \xscaMap(z \MID \lambda)$.
Also, using (\ref{eq:Fdef}),
\beqa
    \lim_{|x|\arr\infty} \varphi(x) &\stackrel{(a)}{=}&
    \lim_{|x|\arr\infty} F(x,z,\lambda)  \nonumber \\
    &\stackrel{(b)}{\geq}&  \lim_{|x|\arr\infty}  f(x)
    \stackrel{(c)}{=} \infty  \label{eq:varphiInfty}
\eeqa
where (a) follows from (\ref{eq:varphiSca});
(b) follows from (\ref{eq:Fdef});
and (c) follows from Assumption~\ref{as:growth}.
Equations (\ref{eq:xscaMapComp}) and (\ref{eq:varphiInfty})
show that item (a) of Lemma~\ref{lem:largeDevPt} is satisfied.
Item (b) of Lemma~\ref{lem:largeDevPt} is also clearly satisfied
since $h(x) = x$ is continuous.

Also, similar to the proof of Lemma \ref{lem:xhatuLim}, one can show
using Assumption~\ref{as:growth} that
item (c) of Lemma~\ref{lem:largeDevPt} is satisfied
for some $c > 0$.  Thus, all the hypotheses of Lemma~\ref{lem:largeDevPt}
are satisfied and we have the limit
\[
    \lim_{u \arr \infty} \xscau(z \MID \lambda) = h(\xscaMap(z \MID \lambda))
    =\xscaMap(z \MID \lambda).
\]
This proves the lemma.
\end{IEEEproof}

We now turn to convergence of the random variable
$\thetaScau(\sigEff^2(u), u\sigPEff^2(u))$.

\begin{lemma} \label{lem:thetaScaLim}
Consider the random vectors
$\thetaScau(\sigma^2,\gamma)$ defined in (\ref{eq:thetaScau})
and $\thetaScaMap(\sigma^2,\gamma)$ in (\ref{eq:thetaScaMap}).
Let $\sigEff^2(u)$, $\sigPEff^2(u)$, $\sigEffMap^2$ and $\gamma_p$
be as defined in Assumption~\ref{as:limitExistence}.
Then the following limit holds:
\beq \label{eq:thetaScaLimLem}
    \lim_{u \arr \infty} \thetaScau(\sigEff^2(u),u\sigPEff^2(u))
    = \thetaScaMap(\sigEffMap^2,\gamma_p)
\eeq
almost surely and in distribution.
\end{lemma}
\begin{IEEEproof}
The proof is similar to that of Lemma~\ref{lem:thetaLimn}.
For any $\sigma^2$ and $\gamma > 0$,
the vectors $\thetaScau(\sigma^2,\gamma)$ and
$\thetaScaMap(\sigma^2,\gamma)$ are deterministic
functions of the random variables $x \sim p_0(x)$,
$s \sim p_S(s)$, and $z$ given (\ref{eq:zxvsca})
with $v \sim {\cal N}(0,1)$.
Lemma~\ref{lem:xhatScauLim} shows that the limit
\beq \label{eq:thetaScaLima}
    \lim_{u \arr \infty} \thetaScau(\sigma^2,\gamma)
    = \thetaScaMap(\sigma^2,\gamma)
\eeq
holds for any values of $\sigma^2$, $\gamma$,
$x$, and $s$ and almost all $z$.
Also, if we fix $x$, $s$, and $v$,
by Assumption~\ref{as:limitingVariance}, the function
\[
    \xscaMap(z \MID \gamma/s) = \xscaMap(x+\frac{\sigma}{\sqrt{s}}v \MID \gamma/s)
\]
is continuous in $\gamma$ and $\sigma^2$ for almost
all values of $v$.
Therefore, we can combine (\ref{eq:thetaScaLima}) with the
limits in Assumption~\ref{as:limitExistence} to show that
\[
    \lim_{u \arr \infty} \thetaScau(\sigEff^2(u),u\sigPEff^2(u))
    = \thetaScaMap(\sigEffMap^2,\gamma_p)
\]
for almost all $x$ and $s$ and almost all $z$.
Since $z$ has a continuous probability distribution
(due to the additive noise $v$ in (\ref{eq:zxvsca})),
the set of values where this limit does not hold must
have probability zero. Thus, the limit (\ref{eq:thetaScaLimLem})
holds almost surely, and therefore, also in distribution.
\end{IEEEproof}

\section{Proof of the Fixed-Point Equations}
\label{sec:fixed-point}
For the final part of the proof, we need to show that
the limits $\sigEffMap^2$ and $\gamma_p$ in Assumption~\ref{as:limitExistence}
satisfy the fixed-point equations (\ref{eq:repMapFix}).
The proof is straightforward,
but we just need to keep track of the notation properly.
We begin with the following limit.

\begin{lemma} \label{lem:mseLim}
The following limit holds:
\beqan
    \lefteqn{ \lim_{u \arr \infty}
        \Exp\left[s \, \mse(p_u,\ptrue,\mu_p^u,\mu^u,z^u)\right] }\\
    &=& \Exp \left[s|x-\xscaMap(z \MID \lambda)|^2\right],
\eeqan
where the expectations are taken over $x \sim \ptrue(x)$ and
$s \sim p_S(s)$,
and $z$ and $z^u$ are the random variables
\begin{subequations}
\beqa
    z^u &=& x + \sqrt{\mu^u}v, \label{eq:zudefLem}\\
    z &=& x + \sqrt{\mu}v, \label{eq:zdefLem}
\eeqa
\end{subequations}
with $v \sim {\cal N}(0,1)$ and
$\mu^u = \sigEff^2(u)/s$, $\mu_p^u = \sigPEff^2(u)/s$,
$\mu = \sigEffMap^2/s$, and $\lambda = \gamma_p/s$.
\end{lemma}
\begin{IEEEproof}  Using the
definitions of $\mse$ in (\ref{eq:mseScalar}) and
$\xscau(z \MID \cdot)$ in (\ref{eq:xscau}),
\beqan
    \lefteqn{\mse(p_u,\ptrue,\mu_p^u,\mu^u,z^u)} \\
     &=&  \int_{x \in \Xset}
            |x-\xscaMmse(z^u \MID p_u,\mu_p^u)|^2
            p_{x \mid z}(x \mid z^u \MID \ptrue,\mu^u) \, dx \nonumber \\
     &=&  \int_{x \in \Xset}
            |x-\xscau(z^u \MID \mu_p^u/u)|^2
            p_{x \mid z}(x \mid z^u \MID \ptrue,\mu^u) \, dx. \nonumber
\eeqan
Therefore, fixing $s$ (and hence $\mu_p^u$ and $\mu^u$),
we obtain the conditional expectation
\beqa
    \lefteqn{\Exp\left[\mse(p_u,\ptrue,\mu_p^u,\mu^u,z^u)
        \mid s\right]} \nonumber \\
    &=& \Exp\left[|x-\xscau(z^u \MID \mu_p^u/u)|^2\mid s\right], \label{eq:mseuLem}
\eeqa
where the expectation on the right is over $x \sim \ptrue(x)$
and $z^u$ given by (\ref{eq:zudefLem}).

Also, observe that the definitions
$\mu^u = \sigEff^2(u)/s$  and $\mu = \sigEffMap^2/s$
and along with the limit in Assumption~\ref{as:limitExistence} show that
\beq \label{eq:mulimLem}
    \lim_{u \arr \infty} \mu^u = \mu.
\eeq
Similarly, since $\mu_p^u = \sigPEff^2(u)/s$
and $\lambda = \gamma_p/s$,
Assumption~\ref{as:limitExistence} shows that
\beq\label{eq:muplimLem}
    \lim_{u \arr \infty} \frac{\mu_p^u}{u} = \lambda.
\eeq
Taking the limit as $u \arr \infty$,
\beqan
    \lefteqn{\lim_{u \arr \infty}
    \Exp\left[s \, \mse(p_u,\ptrue,\mu_p^u,\mu^u,z^u)\right]} \\
    &\stackrel{(a)}{=}& \lim_{u \arr \infty}
        \Exp\left[s |x-\xscau(z^u \MID \mu_p^u/u)|^2\right],\\
    &\stackrel{(b)}{=}& \lim_{u \arr \infty}
        \Exp\left[s |x-\xscau(z^u \MID \lambda)|^2\right],\\
    &\stackrel{(c)}{=}& \lim_{u \arr \infty}
        \Exp\left[s |x-\xscau(z \MID \lambda)|^2\right],\\
    &\stackrel{(d)}{=}& \lim_{u \arr \infty}
        \Exp\left[s |x-\xscaMap(z \MID \lambda)|^2\right],
\eeqan
where (a) follows from (\ref{eq:mseuLem});
(b) follows from (\ref{eq:muplimLem});
(c) follows from (\ref{eq:mulimLem}), which implies that
$z^u \arr z$; and
(d) follows from Lemma~\ref{lem:xhatScauLim}.
\end{IEEEproof}

The previous lemma enables us to evaluate the limit of the MSE in
(\ref{eq:sigEffSolMap}).  To evaluate the limit of the MSE in
(\ref{eq:sigPEffSolMap}), we need the following lemma.

\begin{lemma} \label{lem:gulim}  Fix $z$ and $\lambda$, and let
\beq  \label{eq:guxDef}
    g(u,x) = u|x-\xhat|^2,
\qquad
    \xhat = \xscaMap(z \MID \lambda).
\eeq
Also, let $\varphi(x)$ be given by (\ref{eq:varphiSca})
and $q_u(x)$ be given by (\ref{eq:qu}) with ${\cal D} = \Xset$.
Then, for any $\epsilon > 0$, there exists
an open neighborhood $U \subseteq \Xset$ of $\xhat$
such that
\[
    \limsup_{u \arr \infty} \left| \int_{x \in U} g(u,x)q_u(x) \, dx -
    \sigma^2(z,\lambda) \right| < \epsilon,
\]
where $\sigma^2(z,\lambda)$ is given in Assumption~\ref{as:limitingVariance}.
\end{lemma}
\begin{IEEEproof}
The proof is straightforward but somewhat tedious.
We will just outline the main steps.  Let $\delta > 0$.
Using Assumption~\ref{as:growth}, one can find
an open neighborhood $U \subseteq \Xset$ of $\xhat$
such that for all
$x \in U$ and $u > 0$,
\beq \label{eq:expPhiBnd}
    \phi\left(x,\sigma^2_{-}(u)\right) \leq
    \exp(-u(\varphi(x) - \varphi(\xhat)))
   \leq
    \phi\left(x,\sigma^2_{+}(u)\right),
\eeq
where $\phi(x,\sigma^2)$ is the unnormalized Gaussian
distribution
\[
    \phi(x,\sigma^2) = \exp\left(-\frac{1}{2\sigma^2}|x-\xhat|^2
        \right)
\]
and
\beqan
    \sigma^2_{+}(u) &=& (1+\delta)\sigma^2(z,\lambda)/u, \\
    \sigma^2_{-}(u) &=& (1-\delta)\sigma^2(z,\lambda)/u.
\eeqan
Combining the bounds in (\ref{eq:expPhiBnd}) with
the definition of $q_u(x)$ in (\ref{eq:qu})
and the fact that $U \subseteq \Xset$
shows that for all $x \in U$ and $u > 0$,
\beqan
    q_u(x) &=& \left[\int_{x \in \Xset} e^{-u\varphi(x)} \, dx\right]^{-1}
        e^{-u\varphi(x)} \nonumber \\
        & \leq & \left[\int_{x \in U} \phi(x,\sigma^2_-(u)) \, dx\right]^{-1}
        \phi(x,\sigma^2_+(u)).
\eeqan
Therefore,
\beqa
    \lefteqn{ \int_{x\in U} g(u,x)q_u(x) \, dx
    = \int_{x\in U} u|x-\xhat|^2q_u(x) \, dx } \nonumber \\
   &\leq& \left[\int_{x\in U} \phi(x,\sigma^2_{-}(u)) \, dx\right]^{-1}
   \nonumber  \\
   & & \int_{x\in U} u|x-\xhat|^2\phi(x,\sigma^2_{+}(u)) \, dx.
   \hspace{1in}
   \label{eq:guInt}
\eeqa
Now, it can be verified that
\beq \label{eq:sigMinusLim}
    \lim_{u \arr \infty}
    \int_{x \in U} u^{1/2}\phi(x,\sigma^2_{-}(u)) \, dx
    = \sqrt{2\pi(1-\delta)}\sigma(z,\lambda)
\eeq
and
\beqa
    \lefteqn{ \lim_{u \arr \infty}
    \int_{x \in U} u^{3/2}|x-\xhat|^2\phi(x,\sigma^2_{+}(u)) \, dx }
    \nonumber \\
    &=& \sqrt{2\pi(1+\delta)^3}\sigma(z,\lambda)^3.
    \hspace{1in} \label{eq:sigPlusLim}
\eeqa
Substituting (\ref{eq:sigMinusLim}) and (\ref{eq:sigPlusLim})
into (\ref{eq:guInt}) shows that
\[
    \limsup_{u \arr \infty} \int_{x\in U} g(u,x)q_u(x) \, dx
    \leq \frac{(1+\delta)^{3/2}}{1-\delta}\sigma^2(z,\lambda).
\]
A similar calculation shows that
\[
    \liminf_{u \arr \infty} \int_{x\in U} g(u,x)q_u(x) \, dx
    \geq \frac{(1-\delta)^{3/2}}{1+\delta}\sigma^2(z,\lambda).
\]
Therefore, with appropriate selection of $\delta$,
one can find a neighborhood $U$ of $\xhat$ such that
\[
    \limsup_{u \arr \infty} \left| \int_{x\in U} g(u,x)q_u(x) \, dx
    - \sigma^2(z,\lambda) \right| < \epsilon,
\]
and this proves the lemma.
\end{IEEEproof}

Using the above result, we can evaluate the scalar MSE\@.

\begin{lemma} \label{lem:msepLim}
Using the notation of Lemma~\ref{lem:mseLim},
\[
    \lim_{u \arr \infty} \Exp\left[us \, \mse(p_u,p_u,\mu_p^u,\mu_p^u,z)\right]
    = \Exp\left[s\sigma^2(z,\gamma_p/s)\right].
\]
\end{lemma}
\begin{IEEEproof}  This is an application of Lemma~\ref{lem:largeDev}.
Fix $z$ and $\lambda$ and define $g(u,x)$
as in (\ref{eq:guxDef}).  As in the proof of Lemma
\ref{lem:xhatScauLim}, the conditional distribution
$p_{x\mid z}(x \mid z \MID p_u,\lambda/u)$ is given by (\ref{eq:pxzu})
with $\varphi(x)$ given by (\ref{eq:varphiSca}).
The definition of $\xscaMap(z \MID \lambda)$ in
(\ref{eq:xmapScalar}) shows that $\xscaMap(z \MID \lambda)$
minimizes $\varphi(x)$.  Similar to the proof of
Lemma~\ref{lem:xhatScauLim}, one can show that items (a) and
(b) of Lemma~\ref{lem:largeDev} are satisfied.
Also, Lemma~\ref{lem:gulim} shows that item (c) of
Lemma~\ref{lem:largeDev} holds with $g_\infty = \sigma^2(z,\lambda)$.
Therefore, all the hypotheses of Lemma~\ref{lem:largeDev}
are satisfied and
\beq \label{eq:guMseLim}
 \lim_{u \arr \infty} \int_{x \in \Xset}
    u|x-\xscaMap(z \MID \lambda)|^2q_u(x) \, dx =\sigma^2(z,\lambda),
\eeq
for all $\lambda$ and almost all $z$.

Now
\beqa
   \lefteqn{ \mse(p_u, p_u, \lambda/u, \lambda/u,z) } \nonumber \\
     &\stackrel{(a)}{=}&  \int_{x \in \Xset}
            |x-\xscaMmse(z \MID p_u,\lambda/u)|^2
            p_{x \mid z}(x \mid z \MID p_u,\lambda/u) \, dx \nonumber \\
     &\stackrel{(b)}{=}&  \int_{x \in \Xset}
            |x-\xscaMmse(z \MID p_u,\lambda/u)|^2q_u(x) \, dx  \nonumber \\
     &\stackrel{(c)}{=}&  \int_{x \in \Xset}
            |x-\xscau(z \MID \lambda)|^2q_u(x) \, dx, \label{eq:mseuLim}
\eeqa
where (a) is the definition of $\mse$ in (\ref{eq:mseScalar});
(b) follows from (\ref{eq:pxzu}); and
(c) follows from (\ref{eq:xscau}).
Taking the limit of this expression
\beqa
   \lefteqn{ \lim_{u \arr \infty}
   u \, \mse(p_u, p_u, \lambda/u, \lambda/u,z) } \nonumber \\
   &\stackrel{(a)}{=}& \lim_{u \arr \infty}
   \int_{x \in \Xset}
            u|x-\xscau(z \MID \lambda)|^2q_u(x) \, dx \nonumber \\
   &\stackrel{(b)}{=}& \lim_{u \arr \infty}
   \int_{x \in \Xset}
            u|x-\xscaMap(z \MID \lambda)|^2q_u(x) \, dx \nonumber \\
   &\stackrel{(c)}{=}& \sigma^2(z,\lambda),  \label{eq:mseuLimb}
\eeqa
where (a) follows from (\ref{eq:mseuLim});
(b) follows from Lemma~\ref{lem:xhatScauLim};
and (c) follows from (\ref{eq:guMseLim}).

The variables $z^u$ and $z$ in
(\ref{eq:zudefLem}) and (\ref{eq:zdefLem}) as well as $\mu^u$ and
$\mu^u_p$ are deterministic functions of $x$, $v$, $s$, and $u$.
Fixing $x$, $v$, and $s$ and taking the
limit with respect to $u$ we obtain the deterministic limit
\beqa
 \lefteqn{ \lim_{u \arr \infty}
   u \, \mse(p_u, p_u, \mu_p^u, \mu_p^u,z^u) } \nonumber \\
   &\stackrel{(a)}{=}& \lim_{u \arr \infty}
   u \, \mse(p_u, p_u, \sigPEff^2(u)/s, \sigPEff^2(u)/s,z^u)  \nonumber \\
  &\stackrel{(b)}{=}& \lim_{u \arr \infty} \sigma^2(z^u,u\sigPEff^2(u)/s) \nonumber \\
  &\stackrel{(c)}{=}& \lim_{u \arr \infty} \sigma^2(z,u\sigPEff^2(u)/s) \nonumber \\
  &\stackrel{(d)}{=}& \sigma^2(z,\gamma_p/s), \hspace{2in}
    \label{eq:mseLimuc}
\eeqa
where (a) follows from the definitions of $\mu^u$ and $\mu^u_p$
in Lemma~\ref{lem:mseLim};
(b) follows from (\ref{eq:mseuLimb});
(c) follows from the limit (proved in Lemma~\ref{lem:mseLim})
that $z^u \arr z$ as $u \arr \infty$; and
(d) follows from the limit in Assumption~\ref{as:limitExistence}.

Finally, observe that for any prior $p$ and noise level $\mu$,
\[
    \mse(p, p, \mu, \mu,z) \leq \mu,
\]
since the MSE error must be smaller than the additive noise
level $\mu$.
Therefore, for any $u$ and $s$,
\[
   us \, \mse(p_u, p_u, \mu_p^u, \mu_p^u,z^u) \leq us\mu_p^u
   = u \sigEff^2(u),
\]
where we have used the definition $\mu_p^u = \sigEff^2(u)/s$.
Since $u \sigEff^2(u)$ converges, there must exists a constant $M > 0$
such that
\[
   us \, \mse(p_u, p_u, \mu_p^u, \mu_p^u,z^u) \leq us\mu_p^u
   \leq M,
\]
for all $u$, $s$ and $z^u$.
The lemma now follows from applying the
Dominated Convergence Theorem and taking expectations of both
sides of (\ref{eq:mseLimuc}).
\end{IEEEproof}

We can now show that the limiting noise values satisfy the
fixed-point equations.

\begin{lemma} \label{lem:fixPoint}
The limiting effective noise levels $\sigEffMap^2$
and $\gamma_p$ in Assumption~\ref{as:limitExistence}
satisfy the fixed-point equations (\ref{eq:sigEffSolMap})
and (\ref{eq:sigPEffSolMap}).
\end{lemma}
\begin{IEEEproof}
The noise levels $\sigEff^2(u)$ and $\sigPEff^2(u)$
satisfy the fixed-point equations (\ref{eq:sigEffSol})
and (\ref{eq:sigPEffSol}) of the RS PMMSE decoupling property
with the postulated prior $\ppost = p_u$ and
noise level $\sigpost^2 = \gamma/u$.
Therefore, using the notation in Lemma~\ref{lem:mseLim},
\begin{subequations}
\beqa
    \sigEff^2(u) & = & \sigtrue^2 + \beta \,
        \Exp\left[ s \, \mse(p_u,\ptrue,\mu_p^u,\mu^u,z^u) \right],
        \label{eq:sigMapa} \\
    u\sigPEff^2(u) & = & \gamma + \beta \, \Exp\left[ us \, \mse(p_u,p_u,\mu_p^u,\mu_p^u,z^u)\right], \quad
            \label{eq:sigPMapa}
\eeqa
\end{subequations}
where (as defined in Lemma~\ref{lem:mseLim}),
$\mu^u = \sigEff^2(u)/s$ and $\mu_p^u = \sigPEff^2(u)/s$ and
the expectations are taken over $s \sim p_S(s)$,
$x \sim \ptrue(x)$, and $z^u$ in (\ref{eq:zudefLem}).

Therefore,
\beqan
    \sigEffMap^2
    &\stackrel{(a)}{=}& \lim_{u \arr \infty} \sigEff^2(u) \\
    &\stackrel{(b)}{=}& \sigtrue^2
     + \beta \,
        \Exp\left[ s \, \mse(p_u,\ptrue,\mu_p^u,\mu^u,z^u) \right] \\
    &\stackrel{(c)}{=}& \sigtrue^2
    + \beta \, \Exp \left[s|x-\xscaMap(z \MID \lambda)|^2\right],
\eeqan
where (a) follows from the limit in Assumption~\ref{as:limitExistence};
(b) follows from (\ref{eq:sigMapa}); and
(c) follows from Lemma~\ref{lem:mseLim}.
This shows that (\ref{eq:sigEffSolMap}) is satisfied.

Similarly,
\beqan
    \gamma_p
    &\stackrel{(a)}{=}& \lim_{u \arr \infty} u\sigPEff^2(u) \\
    &\stackrel{(b)}{=}& \gamma
     + \beta \,
        \Exp\left[ s \, \mse(p_u,p_u,\mu_p^u,\mu_p^u,z^u) \right] \\
    &\stackrel{(c)}{=}& \gamma
    + \beta \, \Exp \left[s\sigma^2(z,\lambda_p)\right],
\eeqan
where (a) follows from the limit in Assumption~\ref{as:limitExistence};
(b) follows from (\ref{eq:sigPMapa}); and
(c) follows from Lemma~\ref{lem:msepLim}.
This shows that (\ref{eq:sigPEffSolMap}) is satisfied.
\end{IEEEproof}

\section*{Acknowledgment}
The authors thank reviewers for helpful comments.
They also thank Martin Vetterli and Amin Shokrollahi
for discussions and support.

\end{document}